\PassOptionsToPackage{prologue,dvipsnames}{xcolor}
\documentclass[letterpaper]{scrartcl}
\usepackage[total={16.2cm, 23.5cm}]{geometry}
\usepackage{microtype}
\usepackage{amsmath,amsthm,amssymb, mathtools}
\usepackage{tikz}
\usetikzlibrary{patterns,decorations.pathreplacing,arrows.meta,shapes.geometric}
\usepackage{hyperref}
\hypersetup{
pdfencoding=auto, psdextra,
colorlinks=true,citecolor=green!50!black,linkcolor=red!60!black,
}
\usepackage[sort&compress,nameinlink]{cleveref}
\usepackage{nicefrac}
\usepackage[square,numbers,sort]{natbib}
\usepackage{comment}
\usepackage{paralist}
\usepackage{xcolor}
\usepackage{colortbl}
\usepackage{subcaption}
\usepackage{booktabs}
\usepackage{listings}
\usepackage[linesnumbered,ruled,vlined]{algorithm2e}
\usepackage[linecolor=green!70!black, backgroundcolor=green!10, bordercolor=black, textsize=small]{todonotes}
\definecolor{winered}{rgb}{0.6,0.1,0.1}
\usepackage{pgfplots}\pgfplotsset{compat=newest}
\usepackage{pifont}

\renewcommand*{\leq}{\leqslant}

\renewcommand*{\geq}{\geqslant}
\renewcommand{\epsilon}{\varepsilon}
\newcommand{\pabulib}{{{\textsc{Pabulib}}}}

\crefname{table}{Table}{Tables}
\Crefname{table}{Table}{Tables}
\crefname{figure}{Figure}{Figures}
\crefname{theorem}{Theorem}{Theorems}
\crefname{definition}{Definition}{Definitions}
\crefname{corollary}{Corollary}{Corollaries}
\crefname{observation}{Observation}{Observations}
\crefname{question}{Question}{Question}
\crefname{lemma}{Lemma}{Lemmas}
\crefname{example}{Example}{Examples}
\crefname{reduction}{Reduction}{Reductions}
\crefname{construction}{Construction}{Constructions}
\crefname{subsection}{Section}{Sections}
\crefname{section}{Section}{Sections}
\crefname{proposition}{Proposition}{Propositions}
\crefname{algorithm}{Algorithm}{Algorithms}
\crefname{algocf}{Algorithm}{Algorithms}
\Crefname{equation}{Inequality}{Inequalities}
\crefname{lstlisting}{listing}{listings}

\newcommand{\myemph}[1]{{\color{winered}\emph{#1}}}
\newcommand{\naturals}{{{\mathbb{N}}}}
\newcommand{\feasibles}{{{\mathcal{F}}}}
\newcommand{\reals}{{{\mathbb{R}}}}
\newcommand{\cost}{{{\mathrm{cost}}}}

\newcommand{\best}{{{\mathrm{best}}}}
\newcommand{\money}{{{\mathrm{m}}}}
\newcommand{\avg}{{{\mathrm{avg}}}}

\newcommand{\rem}{{{C_\mathrm{rem}}}}

\newcommand{\spent}{{{\mathrm{spent}}}}

\newcommand{\scaling}{{{\mathrm{scaling}}}}

\newcommand{\nxt}{{{\mathrm{next}}}}

\renewcommand{\part}{{{\mathrm{part}}}}

\theoremstyle{definition}
\newtheorem{definition}{Definition}
\newtheorem{example}{Example}
\newtheorem{lemma}{Lemma}

\theoremstyle{plain}

\newtheorem{theorem}{Theorem}
\newtheorem{corollary}{Corollary}

\makeatletter
\newtheorem*{rep@theorem}{\rep@title}
\newcommand{\newreptheorem}[2]{%
\newenvironment{rep#1}[1]{%
 \def\rep@title{#2 \ref{##1}}%
 \begin{rep@theorem}}%
 {\end{rep@theorem}}}
\makeatother
\newreptheorem{theorem}{Theorem}
\newreptheorem{proposition}{Proposition}
\newreptheorem{lemma}{Lemma}
\newreptheorem{corollary}{Corollary}

\newcommand{\argmin}{{\mathrm{argmin}}}
\makeatletter
\newcommand{\pushright}[1]{\ifmeasuring@#1\else\omit\hfill$\displaystyle#1$\fi\ignorespaces}
\newcommand{\pushleft}[1]{\ifmeasuring@#1\else\omit$\displaystyle#1$\hfill\fi\ignorespaces}
\makeatother

\begin{document}

\title{A General Theory of Proportionality with Additive Utilities}

\author{
Piotr Skowron\\
  University of Warsaw\\
  {\small \url{p.skowron@mimuw.edu.pl}}}
\date{}

\renewcommand{\zeta}{z}

\maketitle

\vspace{-1cm}
\begin{abstract}
We consider a model where a subset of candidates must be selected based on voter preferences, subject to general constraints that specify which subsets are feasible. This model generalizes committee elections with diversity constraints, participatory budgeting (including constraints specifying how funds must be allocated to projects from different pools), and public decision-making. Axioms of proportionality have recently been defined for this general model, but the proposed rules apply only to approval ballots, where each voter submits a subset of candidates she finds acceptable. We propose proportional rules for cardinal ballots, where each voter assigns a numerical value to each candidate corresponding to her utility if that candidate is selected. In developing these rules, we also introduce methods that produce proportional rankings, ensuring that every prefix of the ranking satisfies proportionality.
\end{abstract}

\lstset{
    keywords={input, output, for, while, if, else, return, break},
    comment=[l]{//},
    frame=single,
    mathescape=true,
    float,
    captionpos=b,
    numbers=left,
    breaklines=true,
}

\allowdisplaybreaks
\renewcommand{\contentsname}{\vspace{-2em}}
\tableofcontents

\section{Introduction}

We consider a model in which the objective is to select a subset from a set of available candidates $C$, based on voters' preferences. The selected subset must be \myemph{feasible}, where the feasibility constraints are defined by a family of subsets $\feasibles \subseteq 2^C$; a subset $W \subseteq C$ is a feasible solution if $W \in \feasibles$. This general framework captures a broad range of settings studied in social choice theory~\cite{mas-pie-sko:group-fairness}. For instance, when $\feasibles$ consists of all subsets of a specified cardinality, the model corresponds to committee elections~\cite{FSST-trends,lac-sko:multiwinner-book}. Our goal is to select a feasible subset such that \myemph{voters' preferences are represented in a fair manner}. Fairness is formalized through the notion of proportionality: if a group of voters constituting a fraction $\gamma$ of the electorate shares sufficient agreement on certain candidates, then this group should be able to influence the selection of approximately a $\gamma$ fraction of the chosen candidates. This intuition has been captured by various proportionality axioms. Among the most prominent are variants of \myemph{justified representation}, originally introduced in the context of committee elections~\cite{justifiedRepresentation,pjr17,pet-sko:laminar,kal-liu-kem:full-pjr}.

In this paper, we design \myemph{rules} (algorithms) for selecting proportionally fair outcomes, assuming voters express their preferences through additive utilities---that is, each voter $i$ assigns a utility value $u_i(c)$ to every candidate $c$. The problem of finding and computing proportional outcomes, specifically outcomes satisfying variants of justified representation, is well understood either for very specific types of feasibility constraints (e.g., for committee elections and participatory budgeting) or when voters have simple approval-based preferences, meaning each voter assigns a utility of either 1 or 0 to each candidate.

\subsection*{Feasibility Constraints}

As we have already noted our framework captures the model of committee elections; moreover, it allows to accommodate diversity constraints that impose additional structural requirements on the composition of the elected committee (e.g., that the elected committee should be gender-balanced)~\cite{azi:committees-soft-constraints,conf/aaai/BredereckFILS18,conf/ijcai/CelisHV18}. When $\feasibles$ consists of all subsets whose total cost does not exceed a predetermined budget, we obtain the model of participatory budgeting~\cite{pet-sko:laminar, pet-pie-sko:c:participatory-budgeting-cardinal,rey-mal:survey-on-pb}.

Furthermore, by introducing auxiliary candidates and imposing additional constraints, it is possible to model more complex voting scenarios that do not, at first glance, appear to involve selecting a subset of candidates~\cite{mas-pie-sko:group-fairness}. As an illustrative example, consider committee elections in which voters may approve, disapprove, or remain neutral toward each candidate. One can exploit the expressiveness of the general constrained model to reduce this richer preference setting to the simpler framework in which voters have only non-negative utilities over candidates~\cite{kra-pap-pie-sko:prop-negative-votes}. Specifically, for each candidate $c \in C$, one can introduce an auxiliary candidate $\neg c$ representing the option of \emph{not selecting} $c$. A voter approves $\neg c$ if and only if they disapproved of $c$ in the original election. In this way, voters can derive utility from the exclusion of candidates they oppose, thus still operating within a non-negative utility model. To ensure that this construction behaves as intended, we must impose additional feasibility constraints stating that $c$ and $\neg c$ cannot be selected simultaneously. The same way, our model captures the problem of public decisions where we need to make decisions on a number of independent issues~\cite{conitzer2017fair,fre-kah-pen:variable_committee_elections,sko-gor:proportional_public_decisions,cha-goe-pet:seq-decision-making,lac:perpetual-voting,bulteau2021jr}, the general model of judgement aggregation~\cite{lis-pol:judgment-aggregation,End15JA} or the problem of constructing a schedule of jobs based on the collective preferences of individuals~\cite{sko-rza-pas:collective_scheduling}. For more discussion on how different social choice problems can be modeled via constraints we refer to the recent paper by~\citet{mas-pie-sko:group-fairness}.

Very few rules are known handle even these specific models in the presence of additive utilities. A notable exception is the recent work of \citet{mav-mun-she:committees-with-constraints}, who analyzed a rule based on maximizing a smoothed variant of Nash welfare~\cite{fain2018fair} for general constraints; however, this rule satisfies only weaker notions of proportionality (see the discussion by \citet{mas-pie-sko:group-fairness}) and, in particular, fails variants of justified representation even for participatory budgeting~\cite{pet-sko:laminar}.

In the specific context of participatory budgeting, a rule known as the method of equal shares has been proposed that satisfies extended justified representation (EJR) up-to-one under cardinal utilities~\cite{pet-sko:laminar, pet-pie-sko:c:participatory-budgeting-cardinal,pap-pis-ski-sko-was:bos,BP-ejrplus}. However, this rule is designed specifically for participatory budgeting constraints, and it is not evident how to generalize it to arbitrary feasibility constraints. 

For committee elections with cardinal utilities, a notion stronger than justified representation axioms is often considered, namely the core~\cite{justifiedRepresentation}. Since the core can be empty, considerable effort has been devoted to designing algorithms that satisfy approximate variants of it~\cite{FMS18,mav-mun-she:committees-with-constraints,cheng2019group,jiang2019approx,min-yah-wan:core-approx}, yet these algorithms can in principle fail justified representation axioms quite severely.

\subsection*{Approval Preferences versus Additive Utilities}

Extending the axioms of justified representation to models with general feasibility constraints is conceptually involved. Nevertheless, definitions in this broader setting have recently been proposed~\cite{mas-pie-sko:group-fairness}, and they preserve many of the appealing properties of their original counterparts. However, the rules currently known for computing proportional outcomes under general feasibility constraints are tailored to approval-based preferences~\cite{mas-pie-sko:group-fairness}.

Under approval preferences, each voter simply specifies the set of candidates she approves of. Formally, the approval-based model can be viewed as a special case of the additive-utility model in which a voter assigns utility~1 to each approved candidate and~0 to all others. It remains unclear how rules designed for approval preferences can be extended to richer preference models---in particular, to accommodate additive utilities. Further, some approval-based rules (e.g., Proportional Approval Voting) are known to lose their good properties pertaining to proportionality, when applied to the model with additive utilities in the most natural way~\cite{pet-sko:laminar}.

Even when voters submit approval ballots, in many cases it is essential to use methods that rely on an implicit cardinal utility interpretation. For example, in participatory budgeting it is typically assumed that a voter derives utility from an approved project equal to its cost, and zero utility otherwise~\cite{fal-fli-pet-pie-sko-sto-szu-tal:pb-experiments, bri-for-lac-mal-pet:proportionality-in-abc-pb}. This observation further underscores the importance of developing proportional rules for additive utilities, even when the voters vote using simpler ballot formats.

\subsection*{Algorithms Based on Buying Candidates with Virtual Money}

Our rules view candidate selection as a process of purchasing candidates with virtual currency, and they are inspired by the Phragm\'en sequential rule~\cite{aaai/BrillFJL17-phragmen,skowron:prop-degree,mas-pie-sko:group-fairness,lac-sko:multiwinner-book}. Voters earn virtual money gradually at a constant rate. In the original Phragm\'en rule, whenever voters approving a candidate~$c$ accumulate enough money to cover its cost, they purchase it (their accounts are set to zero), any candidates that would violate feasibility constraints are removed, and the process continues until no unpurchased, non-removed candidate remains.

For additive utilities, the purchasing process is more involved. At any given moment, voters may choose not to purchase an affordable candidate if a more desirable one might become affordable soon. Specifically, each voter performs an internal simulation to estimate what fraction of each available candidate could be purchased at the current time. For each such candidate, the voter calculates the expected utility, scaled by this fraction. This is used to compute the best-case ratio of payment to utility for purchasing a fractional candidate. Voters then use this ratio to determine how much they are willing to spend on specific candidates, ensuring their payment-to-utility ratio does not exceed that of the best fractional purchase. The precise process is described in \Cref{sec:prop_rank}.

This algorithm can also construct a \myemph{proportional ranking} over candidates, that is a ranking, where each prefix, viewed as a committee, is proportional~\cite{proprank}. To achieve this, we simply do not remove infeasible candidates during the purchasing process and rank all candidates by their purchase timestamps. For this reason we call the algorithm \myemph{PropRank}. This addresses another important problem with many applications: constructing proportional rankings in the presence of additive utilities. A similar question has been recently addressed for ranked voters' preferences~\cite{azi-led-pet-pet-rit:prop-rank-ordinal}.

\subsection*{Method of Equal Shares for General Constraints}

The idea behind PropRank can be further developed to generalize the method of equal shares and its variant with bounded overspending to models with general feasibility constraints. The approach is similar to that described previously, but we delay purchasing affordable candidates as long as possible. Specifically, we allow voters to gradually accumulate virtual coins and identify the earliest time at which, if voters were to make their payments, the purchased candidates would violate the feasibility constraints. We stop immediately before this point and purchase candidates at that moment. We remove all infeasible candidates from consideration and continue the process.

We analyze our rules both in terms of theoretical properties and based on simulations performed on real participatory budgeting instances~\cite{fal-fli-pet-pie-sko-sto-szu-tal:pb-experiments}.

\section{The Model}

An \myemph{election} is a tuple $(N, C, \cost, \{u_i\}_{i \in N})$, where:
\begin{enumerate}
\item $N=\{1, 2, \ldots n\}$ is the set of \myemph{voters} and $C = \{c_1, \ldots, c_m\}$ is the set of  \myemph{candidates}.
\item $\cost\colon C \to \naturals$ is a function that associates each candidate with a \myemph{cost}, and 
\item for each $i \in N$, $u_i\colon C \to \reals_{\geq 0}$ is a \myemph{utility function}, which specifies how much voter $i$ appreciates the specific candidates. Without loss of generality, we assume that each candidate gets a positive utility from at least one voter. By $u_{\max}$ we denote the highest utility a voter assigns to a candidate, $u_{\max} = \max_{i \in N, c \in C} u_i(c)$.
\end{enumerate}
If the utilities of all the voters come from the two-element set $\{0, 1\}$, then we call the election \myemph{approval-based}. Another important special case of the model is when the costs of the candidates are all equal; in such a case we speak of \myemph{uniform costs}. We use the notation:
\begin{align*}
A_i = \{c \in C \colon u_{i}(c) > 0\} \qquad \text{and} \qquad N_c = \{i \in N \colon u_{i}(c) > 0\} \text{.}
\end{align*}

\subsection{Feasibility Constraints}

We define \myemph{feasibility constraints} as a family of subsets of candidates, $\feasibles \subseteq 2^C$, that is downward closed---that is, for all $A \subseteq B \subseteq C$, if $B \in \feasibles$, then $A \in \feasibles$.\footnote{The assumption of downward closure is purely technical and does not restrict the generality of our results. One can equivalently define $\feasibles$ as the family of subsets that can be extended to a feasible solution, as explained in the recent work of \citet{mas-pie-sko:group-fairness}.} We refer to the elements of $\feasibles$ as \myemph{feasible subsets}. A \myemph{selection rule} is a function that takes an election and feasibility constraints as input, and returns a feasible subset of the candidates.

A simple example of elections with feasibility constraints is the classic model of committee elections~\cite{lac-sko:multiwinner-book, FSST-trends}, where all subsets of candidates of size at most $k$, for some fixed constant $k$, are feasible. \citet{mas-pie-sko:group-fairness} discuss several other important examples of elections that are naturally defined through feasibility constraints.
One particularly important class of feasibility constraints is that given by matroid constraints.

\begin{definition}[Matroid constraints]
The feasibility constraints $\feasibles$ form a matroid if, for all $X, Y \in \feasibles$ with $|X| < |Y|$, there exists an element $c \in Y \setminus X$ such that $X \cup \{c\} \in \feasibles$. \hfill $\lrcorner$
\end{definition}

Examples of settings that can be modeled using matroid constraints include standard committee elections~\cite{lac-sko:multiwinner-book, FSST-trends}, committee elections with gender quotas~\cite{azi:committees-soft-constraints,conf/aaai/BredereckFILS18,conf/ijcai/CelisHV18,mas-pie-sko:group-fairness}, public decision-making~\cite{sko-gor:proportional_public_decisions, conitzer2017fair, fre-kah-pen:variable_committee_elections}, and sequential decision-making~\cite{lac:perpetual-voting}; see the discussion in the paper by \citet{mas-pie-sko:group-fairness}.

When we state that feasibility constraints have a matroid structure, we implicitly assume that all candidate costs are uniform and equal to one. Indeed, in our context, these costs intuitively correspond to the ``amount of space'' each candidate occupies within the feasibility constraints, and under the matroid assumption, all candidates intuitively carry the same weight.

Another important class of elections corresponds to \myemph{participatory budgeting (PB)}~\cite{participatoryBudgeting, knapsackVoting, pet-pie-sko:c:participatory-budgeting-cardinal}.
\begin{definition}[Participatory budgeting (PB) constraints]
The family $\feasibles$ forms participatory budgeting constraints if there exists a constant $b > 0$, called the \myemph{budget}, such that
\begin{align*}
	\feasibles = \left\{ W \subseteq C \colon \sum_{c \in W} \cost (c) \leq b \right\} \text{.} \tag*{\llap{$\lrcorner$}}
\end{align*} 
\end{definition}

\subsection{Proportionality Axioms}

The proportionality axioms that we consider are all based on the idea of a cohesiveness. Below provide a definition that extends the two provided by \citet{mas-pie-sko:group-fairness} and \citet{pet-pie-sko:c:participatory-budgeting-cardinal}. 

\begin{definition}[Cohesiveness for general feasibility constraints]\label{def:cohesiveness}
Consider an election $E = (N, C, \cost, \{u_i\}_{i \in N})$. We say that a group of voters $S$ is $(\alpha, \beta)$-\myemph{cohesive}, $\alpha \in \reals$, $\beta \in \naturals$, if for each feasible set $T \in \feasibles$ with size bounded by
\begin{align*}
|T| < \beta \cdot \frac{n - |S|}{|S|} \text{,}
\end{align*}
there exists a set $X$ such that $|X| \leq \beta$, $\sum_{c \in X} \min_{i \in S}u_i(c) \geq \alpha$, and $X \cup T \in \feasibles$. A group of voters $S$ is $\alpha$-\myemph{cohesive} if it is $(\alpha, \beta)$-\myemph{cohesive} for some $\beta \in \naturals$.
\hfill $\lrcorner$
\end{definition}

The definition is somewhat involved, but at a high level, it can be summarized as follows. Intuitively, a group $S$ agrees that a candidate $c$ represents a value of at least $\min_{i \in S} u_i(c)$; thus, the sum $\sum_{c \in X} \min_{i \in S} u_i(c)$ can be viewed as a certain strong form of agreement among the voters in $S$ with respect to the value of $X$. The group $S$ is considered $(\alpha, \beta)$-cohesive if, intuitively, it has the right to select $\beta$ candidates on which $S$ agrees their value is at least $\alpha$.
The key question is: when should such a group be granted the right to $\beta$ candidates? To establish this, consider the voters from $N \setminus S$. In the worst case from $S$’s perspective, these voters might have preferences disjoint from those of $S$, leading to a potential conflict with respect to which candidates should be selected. If $S$ is entitled to $\beta$ candidates, then, by proportionality, the remaining voters would be entitled to roughly $\beta \cdot \nicefrac{n - |S|}{|S|}$ candidates. 
Thus, $S$ should have the right to $\beta$ candidates if their selection would not reduce the number of candidates available to $N \setminus S$ (within feasibility constraints) below $\beta \cdot \nicefrac{n - |S|}{|S|}$; formally, if such a selection would not deprive the voters from $N \setminus S$ of the right to buy a set $T$ of size lower than $\beta \cdot \nicefrac{n - |S|}{|S|}$.

A more specific definition for participatory budgeting constrains follows~\cite{pet-pie-sko:c:participatory-budgeting-cardinal}:

\begin{definition}[Cohesiveness for participatory budgeting constraints with the budget $b$]\label{def:pb_constraints}
We say that a group of voters $S$ is $\alpha$-\myemph{cohesive}, $\alpha \in \reals$, if there exists $X \subseteq C$ such that the voters from $S$ given a proportional share of the budget $b$ can afford to buy candidates from $X$ (that is
$b \cdot \nicefrac{|S|}{n} \geq \cost(X)$), and if $\sum_{c \in X} \min_{i \in S}u_i(c) \geq \alpha$. \hfill $\lrcorner$
\end{definition}

In this paper we introduce a further strengthening of the condition of cohesiveness in \Cref{def:cohesiveness} by requiring that the voters must agree on the utility of all candidates in $X$ simultaneously. This leads us to a stronger notion of cohesiveness, and as a result to a weaker notion of proportionality. 

\begin{definition}
    We say that a group of voters $S$ is strongly $(\alpha, \beta)$-\myemph{cohesive}, $\alpha \in \reals$, $\beta \in \naturals$, if for each feasible set $T \in \feasibles$ with size bounded by $|T| < \beta \cdot \nicefrac{n - |S|}{|S|}$, there exists a set $X$ with $|X| = \beta$, $\beta \cdot \min_{i \in S, c \in X}u_i(c) \geq \alpha$, and $X \cup T \in \feasibles$. \hfill $\lrcorner$
\end{definition}

One of the most important proportionality axioms is extended justified representation~\cite{justifiedRepresentation,mas-pie-sko:group-fairness,pet-pie-sko:c:participatory-budgeting-cardinal}.

\begin{definition}[Extended Justified Representation]
Given an election $E = (N, C, \cost, \{u_i\}_{i \in N})$ we say that an outcome $W$ satisfies \myemph{extended justified representation (EJR)} if for each $\alpha \in \reals$, and $\alpha$-cohesive group of voters $S$ there exists $i \in S$ such that $u_i(W) \geq \alpha$. An outcome $W$ satisfies \myemph{weak extended justified representation (EJR)} if for each $\alpha \in \reals$, and strongly $\alpha$-cohesive group of voters $S$ there exists $i \in S$ such that $u_i(W) \geq \alpha$. We say that a selection rule satisfies EJR (respectively, weak EJR) if it always returns outcomes that satisfy EJR (resp., weak EJR).  \hfill $\lrcorner$
\end{definition}

It is worth noting that for approval ballots cohesiveness and strong cohesiveness are equivalent, which also implies equivalence of EJR and weak EJR for approval ballots.

\subsection{Ranking Rules}

A \myemph{ranking rule} is a function that, given an election as input, produces a ranking---i.e., a linear order---of the candidates. Since the ranking rule orders all available candidates, it does not need to consider feasibility constraints. To define the proportionality of ranking rules, we follow the approach outlined by \citet{proprank}.

\begin{definition}[Extended Justified Representation for Ranking Rules]
Given an election $E = (N, C, \cost, \{u_i\}_{i \in N})$ we say that a ranking of candidates $\succ$ satisfies extended justified representation if each prefix of the ranking $W$ satisfies EJR for participatory budgeting constraints with the budget set to $b = \sum_{c \in W}\cost(c)$.
We say that a ranking rule satisfies EJR if it always returns rankings that satisfy EJR.  \hfill $\lrcorner$
\end{definition}

\section{Proportionality Degree with Additive Utilities}

In the definition of EJR, we require the existence of at least one satisfied voter. However, this definition implicitly ensures that the overall utility within the cohesive group is high. When we remove from a group $S$ the voter $i$ who is guaranteed to have high utility, we obtain a smaller cohesive group $S \setminus {i}$ to which the axiom still applies (though it may imply an appropriately smaller guarantee for this reduced group). By iteratively applying this reasoning, we can deduce that the entire group $S$ is guaranteed to have high utility. The axiom of proportionality degree implements this intuition more explicitly, and provides a means to quantify the extent to which the preferences of the group of voters are considered. Proportionality degree was first proposed in the context of approval preferences~\cite{skowron:prop-degree}; below we generalize it for additive utilities.

\begin{definition}[Proportionality Degree]\label{def:ejr_generalization}
Consider an election $E = (N, C, \cost, \{u_i\}_{i \in N})$, and a subset of candidates $W \subseteq C$, $\cost(W) \leq b$. We say that $W$ has the \myemph{proportionality degree} of $d\colon \reals \to \reals$ if for each $\alpha \in \reals$, and $\alpha$-cohesive group of voters $S$ 
the average satisfaction of these voters equals to at least $d\left(\alpha\right)$:
\begin{align*}
\frac{1}{|S|} \sum_{i \in S} u_i(W) \geq d\left(\alpha\right) \text{.}
\end{align*}
A selection rule has the proportionality degree of $d$ if it always returns outcomes that have the proportionality degree of $d$.  \hfill $\lrcorner$
\end{definition}

Analogously, we formulate the definition of weak proportionality degree by simply substituting the condition of cohesiveness with strong cohesiveness in \Cref{def:ejr_generalization}.

Now, we establish the relation between proportionality degree and EJR, and show that this relation is essentially the same as in the special case of approval-based committee elections.

\begin{theorem}\label{thm:prop_degree_ejr}
Consider participatory budgeting constraints, and assume that the number of voters is divisible by the budget. If a committee $W$ satisfies EJR, then it has the proportionality degree of $d(\alpha) = (\alpha - u_{\max})/2$, where  $u_{\max}$ is the highest utility a voter assigns to a candidate.
\end{theorem}
\begin{proof}
Consider an election $E = (N, C, \cost, \{u_i\}_{i \in N})$, a budget value $b$, and a subset of candidates $W \subseteq C$ that satisfies EJR. We will prove the statement of the theorem by induction on the size of cohesive groups. For cohesive groups of size zero the statement clearly holds. Let $S \subseteq N$ be an $\alpha$-cohesive group of voters, and let us assume that the statement holds for all cohesive groups of size strictly smaller than $|S|$; we will prove it for $S$. Let $X$ be the set that witnesses that $S$ is cohesive. For each $c' \in X$ let $\alpha(c') = \min_{i \in S} u_i(c')$, and
for each $c' \in X$ we define $s(c')$ as:
\begin{align*}
s(c') = \frac{n}{b} \cdot (\cost(X) - \cost(c')) \text{.} 
\end{align*}
Now, observe that we have:
\begin{align*}
1 - \frac{s(c')}{|S|} \geq \left(1 - \frac{s(c')}{\cost(X) \cdot \frac{n}{b}}\right) = \left(1 - \frac{\frac{n}{b} \cdot (\cost(X) - \cost(c'))}{\cost(X) \cdot \frac{n}{b}}\right) 
                 = \frac{\cost(c')}{\cost(X)} \text{.}
\end{align*}
By summing up this expression over all $c' \in X$, we get:
\begin{align*}
\sum_{c' \in X} \frac{|S| - s(c')}{|S|} \geq \sum_{c' \in X} \frac{\cost(c')}{\cost(X)} = 1 = \sum_{c' \in X} \frac{\alpha(c')}{\sum_{c \in X}\alpha(c)}\text{.}
\end{align*}
Consequently, there must exist $c'$ such that 
\begin{align}\label{eq:pigeonhole_prop_degree_implication}
\frac{|S| - s(c')}{|S|} \geq \frac{\alpha(c')}{\sum_{c \in X} \alpha(c)} \text{.}
\end{align}
Let $X' = X \setminus \{c'\}$. Observe that a group $S' \subseteq S$ with $|S'| \geq s(c')$ is $(\alpha - \alpha(c'))$-cohesive.
Thus, by EJR, in each such a group there exists a voter $i$ such that $u_i(W) \geq \sum_{c \in X'} \alpha(c)$. We apply this reasoning to $S' = S$, finding a voter $i$ with appropriately high utility, next we apply it to $S' = S \setminus \{i\}$, etc., until we end up with a group of size $s' = s(c')$. Now, using the inductive assumption we can assess the total utility of the voters from $S$:
\begin{align*}
\sum_{i \in S} u_i(W) &\geq \underbrace{(|S| -  s') \sum_{c \in X'} \alpha(c)}_{\text{EJR applied several times}} +  \underbrace{\frac{1}{2} \cdot s' \cdot \left(\sum_{c \in X'} \alpha(c) - u_{\max} \right)}_{\text{inductive assumption}} \\
&\geq  \frac{1}{2} \cdot |S| \left(\sum_{c \in X} \alpha(c) - u_{\max}\right) + \frac{1}{2} \cdot |S|(u_{\max} - \alpha(c')) + \frac{1}{2} \cdot |S| \cdot \sum_{c \in X'} \alpha(c) \\
&\qquad\qquad - s' \sum_{c \in X'} \alpha(c) + \frac{1}{2} \cdot s' \cdot \left(\sum_{c \in X'} \alpha(c) -u_{\max}\right) \\
&\geq \frac{1}{2} \cdot |S| \left(\sum_{c \in X} \alpha(c) - u_{\max}\right) + \frac{1}{2} \cdot |S|(u_{\max} - \alpha(c')) + \frac{1}{2} \cdot |S| \cdot \sum_{c \in X'} \alpha(c) \\
&\qquad\qquad - \frac{1}{2} \cdot s' \cdot \left(\sum_{c \in X'} \alpha(c) + u_{\max}\right)  \\
&\geq \frac{1}{2} \cdot |S| \left(\sum_{c \in X} \alpha(c) - u_{\max}\right) - \frac{1}{2} \cdot |S|\alpha(c') + \frac{1}{2} \cdot |S| \cdot \sum_{c \in X'} \alpha(c) \\
&\qquad\qquad - \frac{1}{2} \cdot s' \cdot \sum_{c \in X'} \alpha(c) + \frac{1}{2} \cdot u_{\max} \cdot (|S| - s') \\
&\geq \frac{1}{2} \cdot |S| \left(\sum_{c \in X} \alpha(c) - u_{\max}\right) - \frac{1}{2} \cdot |S|\alpha(c') + \frac{1}{2} \cdot (|S| - s')\cdot \sum_{c \in X'} \alpha(c) \\
&\qquad\qquad + \frac{1}{2} \cdot u_{\max} \cdot (|S| - s') \\
&\geq \frac{1}{2} \cdot |S| \left(\sum_{c \in X} \alpha(c) - u_{\max}\right) \underbrace{- \frac{1}{2} \cdot |S|\alpha(c') + \frac{1}{2} \cdot (|S| - s')\cdot \sum_{c \in X} \alpha(c)}_{\text{we can estimate this by \eqref{eq:pigeonhole_prop_degree_implication}}} \\
&\geq \frac{1}{2} \cdot |S| \left(\sum_{c \in X} \alpha(c) - u_{\max}\right) \text{.}
\end{align*}  
This completes the inductive step, hence the proof.
\end{proof}

Note that \Cref{thm:prop_degree_ejr} provides an implication that applies to each cohesive group of voters $S$. By examining the proof of the theorem, we can see that $u_{\max}$ actually corresponds to the highest utility a voter in $S$ assigns to a candidate within a set that demonstrates the group's cohesiveness. Although this formulation is stronger, we prefer the original for its simplicity. 

An analogous implication holds for approval-based elections with matroid constraints~\cite{mas-pie-sko:group-fairness}. \Cref{ex:ejr_lack_of_implication} shows an instance with matroid constraints but with additive utilities where the implication does not hold. In \Cref{thm:prop_degree_ejr_matroid} however we show a weaker implication that depends on the size of the cohesive group, and applies to all types of constraints.

\begin{example}\label{ex:ejr_lack_of_implication}
We will present an algorithm for constructing an instance in which EJR does not guarantee a proportionality degree of $d(\alpha) = (\alpha - u_{\max})/2$. Fix a natural number $z \in \naturals$, and consider an instance of public decisions, where the set of candidates is divided into disjoint pairs $C = C_1 \cup \ldots C_{2z}$, $|C_i| = \{a_i, b_i \}$ for each $i \in \{1, \ldots, 2z\}$; the feasible outcomes are those that involve at most one candidate from each pair. The number of voters equals to $n = z$. Consider a group of voters $S$ with $|S| = z$ who unanimously assign the utility of $(z - i)^2$ to $a_i$ for $i \leq z$. Additionally, the $i$-th voter from $S$ assigns the utility of $w_i$ (specified later on) to $a_{z + i}$. Consider a subset $S' \subseteq S$, and let us estimate the highest value of $\alpha$ such that $S'$ is $(\alpha, \beta)$-cohesive for some fixed $\beta\in \naturals$. For that our algorithm considers each possible value of $\beta \in \{1, \ldots z\}$. For a fixed value of $\beta$ we construct an adversarial set $T$, such that:
\begin{align*}
|T| = \left\lfloor\beta \cdot \frac{n - |S'|}{|S'|} \right\rfloor \text{.}
\end{align*} 
In the worst case, $T$ consists of $|T|$ candidates $b_i$ with the lowest indices. Then, the best corresponding set $X$ gives the utility of at most:
\begin{align*}
\sum_{i = |T|}^{\min(z, \beta + |T|)} (z - i)^2 \text{.}
\end{align*}
The maximum value of the above expression over all $\beta$ is the previously mentioned utility $w_{|S'|}$. 
Thus, the solution that consists of all $b_i$ for $i \leq z$ and $a_i$ with $i > z$ satisfies EJR. The average of utilities of the voters from $S$ is thus:
\begin{align*}
\frac{1}{z} \cdot \sum_{i = 1}^{z} w_i \text{.}
\end{align*}
The original group $S$ is at least $\alpha$-cohesive for:
\begin{align*}
\alpha_S =  \sum_{i = 1}^{z} (z - i)^2  \text{.}
\end{align*}
The value of the expression $(u_S + z^2)/\alpha_S$ for $z = 200$ equals to approximately $0.38$ which is lower than $0.5$. \hfill $\lrcorner$
\end{example}

\begin{theorem}\label{thm:prop_degree_ejr_matroid}
If an outcome $W$ satisfies EJR, then for each $(\alpha, \beta)$-cohesive group of voters $S$, the average utility from $W$ within the group $S$ equals at least:
\begin{align*}
\frac{1}{|S|} \sum_{i \in S} u_i(W) \geq \alpha \cdot  (\gamma - 1) \cdot \left(\gamma \cdot \log\left( \frac{\gamma}{\gamma-1}\right) - 1 \right) - u_{\max}
\end{align*}
where $\gamma = \nicefrac{n}{|S|}$, and $u_{\max}$ is the highest utility a voter assigns to a candidate.
\end{theorem}
\begin{proof}
Let $S \subseteq N$ be an $(\alpha, \beta)$-cohesive group of voters. 
First, we will show that each group $S' \subseteq S$ is $(\alpha', \beta')$-cohesive for
\begin{align*}
 \beta' = \beta \cdot \frac{(n - |S|)|S'|}{(n - |S'|)|S|}       \qquad     \alpha' = \frac{ \alpha}{\beta} \cdot \lfloor \beta' \rfloor \geq \frac{ \alpha}{\beta} \cdot \beta' - u_{\max}\text{.}
\end{align*}

Indeed, consider a feasible set of candidates $T \in \feasibles$. If there exists a set $X$ with  $|X| \leq \beta$, $\sum_{c \in X} \alpha(c) \geq  \alpha$, and $X \cup T \in \feasibles$, then we can simply pick $\lfloor \beta' \rfloor$ candidates $c \in X$ with the highest values of $\alpha(c)$, and the condition of cohesiveness is satisfied. Otherwise, by the fact that $S$ is $(\alpha, \beta)$-cohesive we know that:
\begin{align*}
|T| \geq \beta \cdot \frac{n - |S|}{|S|} \geq \beta \cdot \frac{(n - |S|)|S'|}{(n - |S'|)|S|}  \cdot \frac{n - |S'|}{|S'|} = \beta' \cdot \frac{n - |S'|}{|S'|} \text{,}
\end{align*}
and so, the condition for cohesiveness is also satisfied.

Now, similarly, as in the proof \Cref{thm:prop_degree_ejr} we will apply EJR multiple times, to different subsets of $S$. As the result, we get that:
\begin{align*}
\sum_{i \in S} u_i(W) &\geq \sum_{s' = 0}^{|S|}   \frac{ \alpha}{\beta} \cdot \lfloor \beta' \rfloor \geq \sum_{s' = 0}^{|S|}  \alpha \cdot \frac{(n - |S|)s'}{(n - s')|S|} - |S| \cdot u_{\max} \\
                      &= \alpha \cdot \frac{n - |S|}{|S|} \cdot \sum_{s' = 1}^{|S|} \frac{s' - n + n}{n - s'} - |S| \cdot u_{\max} \\
		     &= \alpha \cdot \frac{n - |S|}{|S|} \cdot \sum_{s' = 1}^{|S|} \left(\frac{n}{n - s'} - 1\right) - |S| \cdot u_{\max} \\
		     &= \alpha \cdot \frac{n(n - |S|)}{|S|} \cdot \sum_{s' = 1}^{|S|} \frac{1}{n - s'} - \alpha (n - |S|) - |S| \cdot u_{\max} \\
		     &\geq \alpha \cdot \frac{n(n - |S|)}{|S|} \cdot \log\left( \frac{n}{n-|S|}\right) - \alpha (n - |S|) - |S| \cdot u_{\max} \\
		     &= |S| \alpha \cdot \frac{|S|\gamma (|S|\gamma - |S|)}{|S|^2} \cdot \log\left( \frac{|S|\gamma}{|S|\gamma-|S|}\right) - \alpha (|S|\gamma - |S|) - |S| \cdot u_{\max} \\
		     &= |S| \alpha \cdot \left( \gamma (\gamma - 1) \cdot \log\left( \frac{\gamma}{\gamma-1}\right) -  (\gamma - 1) \right) - |S| \cdot u_{\max} \text{.}
\end{align*}
This completes the proof.
\end{proof}
In order to better understand the formula of proportionality degree in \Cref{thm:prop_degree_ejr_matroid} we depict it in \Cref{fig:prop_degree_general}. We can see that the function quickly approaches $\nicefrac{1}{2}$ which is the theoretical guarantee for participatory budgeting constraints (\Cref{thm:prop_degree_ejr}) and approval-based elections with matroid constraints~\cite{mas-pie-sko:group-fairness}.

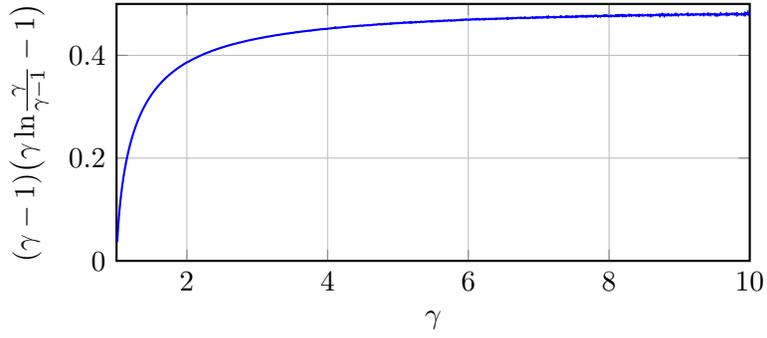
\begin{figure}[t!]
\begin{center}
\begin{tikzpicture}
  \begin{axis}[
    domain = 1.01:10,
    samples = 1000,
    xlabel = {$\gamma$},
    ylabel = {$(\gamma - 1)\bigl(\gamma \ln\!\frac{\gamma}{\gamma - 1} - 1\bigr)$},
    xmin = 1, xmax = 10,
    ymin = 0, ymax = 0.5,
    grid  = both,
    thick,
    height = 5cm,  
    width = 10cm
  ]
    \addplot[blue]
      {(x-1)*(x*ln(x/(x-1)) - 1)};
  \end{axis}
\end{tikzpicture}
\end{center}
\caption{The plot shows the proportionality degree implied by the EJR axiom as a function of the size of the cohesive group of voters~$S$, with $\gamma = n / |S|$. It illustrates the value of the proportionality coefficient (the coefficient next to $\alpha$ in the proportionality degree) which, in the case of participatory budgeting constraints and approval-based elections with matroid constraints, equals to $\nicefrac{1}{2}$. In general case the coefficient converges to~$\nicefrac{1}{2}$ in the limit. The coefficient is approximately $0.39$ and $0.43$ for $\gamma = 2$ and $\gamma = 3$, respectively.}\label{fig:prop_degree_general}
\end{figure}

Further, a stronger implication holds for the weaker variant of EJR.

\begin{theorem}\label{thm:matroid_weak_ejr}
Consider matroid constraints. If a committee $W$ satisfies weak EJR, then it has weak proportionality degree of $d(\alpha) = (\alpha - u_{\max})/2$, where  $u_{\max}$ is the highest utility a voter assigns.
\end{theorem}
\begin{proof}
The proof follows a similar strategy to the proofs of \Cref{thm:prop_degree_ejr} and \Cref{thm:prop_degree_ejr_matroid}. Let $S \subseteq N$ be an $(\alpha, \beta)$-cohesive group of voters. 
First, we will show that each group $S' \subseteq S$ is $(\alpha', \beta')$-cohesive for
\begin{align*}
 \beta' = \left\lfloor \beta \cdot \frac{|S'|}{|S|}\right\rfloor    \qquad     \alpha' = \frac{ \alpha}{\beta} \cdot \beta' \text{.}
\end{align*}

Consider a feasible set $T \in \feasibles$ with $|T| < \beta' \cdot \frac{n - |S'|}{|S'|}$. From $T$ we remove arbitrary $\left\lceil\beta \cdot \frac{|S| - |S'|}{|S|}\right\rceil$ we obtain a set $T'$ such that:
\begin{align*}
|T'| &= |T| - \left\lceil\beta \cdot \frac{|S| - |S'|}{|S|} \right\rceil \leq |T| - \beta \cdot \frac{|S| - |S'|}{|S|} \\
     &< \beta' \cdot \frac{n - |S'|}{|S'|} - \beta \cdot \frac{|S| - |S'|}{|S|} =  \beta \cdot \frac{n - |S'|}{|S|} - \beta \cdot \frac{|S| - |S'|}{|S|} = \beta \cdot \frac{n - |S|}{|S|} \text{.}
\end{align*}

Since $S$ is cohesive, we know there exists a set $X$ with  $|X| \leq \beta$, $\beta \min_{c \in X} \min_{i \in S} u_i(c) \geq \alpha$, and $X \cup T' \in \feasibles$. Since the feasibility constraints form a matroid we know that there exists a set $X' \subseteq X$ with $|X'| = |X| - \left\lceil\beta \cdot \frac{|S| - |S'|}{|S|}\right\rceil$ such that $X' \cup T \in \feasibles$. Note that:
\begin{align*}
|X'| = |X| - \left\lceil\beta \cdot \frac{|S| - |S'|}{|S|}\right\rceil = \left\lfloor\beta - \beta \cdot \frac{|S| - |S'|}{|S|}\right\rfloor = \left\lfloor\beta \cdot \frac{|S'|}{|S|}\right\rfloor = \beta'\text{.}
\end{align*}
Further, the utility that $S'$ agrees on with respect to $X$ is at least:
\begin{align*}
\beta' \min_{c \in X'} \min_{i \in S} u_i(c) \geq \beta' \cdot \frac{\alpha}{\beta} \geq \left\lfloor \beta \cdot \frac{|S'|}{|S|}\right\rfloor \cdot \frac{\alpha}{\beta} \text{.}
\end{align*}
This shows that $S'$ is indeed $(\alpha', \beta')$-cohesive.

Now, similarly, as in the proof \Cref{thm:prop_degree_ejr} we will apply EJR multiple times, to different subsets of $S$. As the result, we get that:
\begin{align*}
\sum_{i \in S} u_i(W) &\geq \sum_{i = 0}^{|S|} \left\lfloor \beta \cdot \frac{i}{|S|}\right\rfloor \cdot \frac{\alpha}{\beta} \geq \sum_{i = 0}^{\beta-1} i \cdot \frac{|S|}{\beta} \cdot \frac{\alpha}{\beta} 
            = \frac{|S|}{\beta} \cdot \frac{\alpha}{\beta} \cdot \frac{\beta(\beta-1)}{2}\text{.}
\end{align*}
Thus, the average utility withom the group $S$ equals to at leats:
\begin{align*}
\alpha \cdot \frac{(\beta-1)}{2\beta} = \frac{\alpha}{2} - \frac{\alpha}{2\beta} \geq \frac{\alpha}{2} - \frac{u_{\max}}{2}\text{.}
\end{align*}
This completes the proof.
\end{proof}

\section{The PropRank Rule}\label{sec:prop_rank}

In this section, we present our main algorithm. We begin with its basic version and in the subsequent sections we discuss its several enhancements. The core idea is inspired by the Phragm\'en rule, in which voters continuously accumulate credits and collectively purchase approved candidates as soon as they can afford them. In our setting, however, candidates may provide different levels of utility to different voters, which introduces additional complexity. In particular, it may be more profitable for a voter to refrain from purchasing a currently affordable candidate and instead save their credits for a potentially more cost-effective one---that is, a candidate with a lower payment-per-utility ratio. This idea is reminiscent of how the method of equal shares proceeds. In our algorithm, voters attempt to anticipate whether a more preferable candidate is likely to become affordable in the future, and make their purchasing decisions accordingly.

\subsection{General Scheme}\label{sec:proprank_general_scheme}

The rule is defined as a process of purchasing the candidates over time. We assume that the voters start with empty budgets and each of them earns money with a constant speed of one dollar per time unit. Thus, at time $t$ each voter earned the total amount of $t$ dollars and spent at most $t$ dollars in total on buying some candidates. The rule works in rounds, starting with an empty outcome $W = \emptyset$, and  in each round buying one candidate that is available, and adding it to $W$. We say that a candidate $c$ is \myemph{available} given an already selected set of candidates $W$ if $c \notin W$, and $W \cup \{c\} \in \feasibles$. Each round takes place in a certain time, and multiple rounds can happen at the same time. 

Let $p_i(r)$ denote the amount of money that voter $i$ holds in round $r$; let $t(r)$ denote the time of round $r$. Clearly, $p_i(r)$ can be lower than $t(r)$ since the voter could already make some purchases in previous rounds. Now a new element of the algorithm appears, where we assume that the voter agrees to spend only a certain fraction of her money on certain candidates. The algorithm is parameterized by $\kappa \in [0, 1]$, which controls how much each voter is willing to spend on a candidate. The main properties of the algorithm are preserved for all values of $\kappa$, although we will later discuss the most recommended choice. Specifically, in round~$r$, a voter $i$ agrees to spend at most $x_i(c, r)$ dollars on a candidate $c$, where:
\begin{align}\label{eq:caps}
x_i(c, r) = \kappa \cdot \frac{2p_i(r)u_i(c)}{u_{i}(c) + \max(\lambda_i(r), u_i(c))} + (1 - \kappa) \cdot \frac{p_i(r)u_i(c)}{\max(\lambda_i(r), u_i(c))} \text{.} 
\end{align}
where $\lambda_i(r)$ is a scaling factor, discussed below.
Note that $x_i(c, r)$ cannot exceed $p_i(r)$, the total amount of money the voter holds at round $r$.
Whenever there exists a candidate $c$ that is affordable, meaning $\sum_{i \in N} x_i(c, r) \geq \cost(c)$, then it is considered for purchase. Specifically, we compute a minimal value of $\rho$ such that 
\begin{align*}
\sum_{i \in N} \min\big(x_i(c, r), u_i(c) \cdot \rho \big) = \cost(c) \text{,}
\end{align*}
select a candidate considered for purchase that corresponds to this value of $\rho$, and adjust voters budgets, by setting
\begin{align*}
p_i(r) := p_i(r) - \min\big(x_i(c, r), u_i(c) \cdot \rho \big) \text{.}
\end{align*}
The cnadidate is included in $W$, and all candidates $c'$ such that $W \cup \{c'\} \notin \feasibles$ are removed.

\subsection{Scaling Factors}\label{sec:scaling_factors}

For each voter $i \in N$ we keep a value, called \myemph{global spending factor} and denoted as $g_i(r)$. The value of the spending factor is initially set to zero, and then it is updated in each round.
Before we explain how it is computed we need to introduce some additional notation. We say that a candidate $c$ is $(\mu, S, \alpha)$-affordable, $\mu, \alpha \in \reals$, $S \subseteq N$, if all voters from $S$ value $c$ at least as $\alpha$, and have enough money to buy the $\mu$ fraction of $c$. Formally, $c$ is $(\mu, S, \alpha)$-affordable if:
\begin{align*}
\sum_{i \in S} p_i(r) = \mu \cdot \cost(c) \qquad \text{and} \qquad \min_{i \in S} u_i(c) \geq \alpha \text{.}
\end{align*}
The idea of considering partial purchases of the candidates is loosely inspired by the mechanism of the method of equal shares with bounded overspending~\cite{pap-pis-ski-sko-was:bos}.

For each voter $i$ we compute the maximal value of $\mu\alpha$, where the maximum is taken over all $(\mu, S, \alpha)$-affordable and available candidates $c'$ with $i \in S$. Then the scaling factor $\lambda_i(r)$ is set to maximum from the aforementioned maximal value of $\mu\alpha$ and $g_i(r)$. If a candidate $c$ is purchased, then for each voter $i \in N$ its global spending factor $g_i$ is updated to the maximum of $g_i$ and the scaling factor from the current round. In other words, global scaling factors are introduced in order to ensure that the scaling factors for each voter are monotonic in consecutive rounds.

The pseudo-code of the algorithm is provided in \Cref{alg:base_prop_rank_algorithm}. The algorithm uses discretization, increasing time in fixed increments of $\Delta_t$. While the exact calculation---determining the earliest time a candidate can be purchased---is achievable in polynomial time, the discretized version is easier to illustrate the main idea of the algorithm.  
Our implementation employs binary search to determine the earliest time a candidate can be purchased, which makes it run significantly faster than computing the time exactly or than using fixed time increments.

\SetKw{Return}{return}
\SetKw{Input}{Input:}
\begin{algorithm}[!thb]
	\small
	\captionsetup{labelfont={sc,bf}, labelsep=newline}
	\caption{Pseudo-code of the PropRank Rule.}\label{alg:base_prop_rank_algorithm}
	\Input{an election with feasibiity constraints $\feasibles$\\}
	\DontPrintSemicolon
	\SetAlgoNoEnd
	\SetAlgoLined
	$W \gets \emptyset$ \; 
	$\rem \gets C$ {\color{winered}\tcc*[l]{remaining candidates}}
	$\spent_i \gets 0  \text{~for each~} v_i \in V$  {\color{winered}\tcc*[l]{money spent by each voter}}
	$g_i \gets 0  \text{~for each~} v_i \in V$  {\color{winered}\tcc*[l]{global scaling factors}}
	$t \gets 0$   {\color{winered}\tcc*[l]{time}}
	\While{$\rem \neq \emptyset$}{
		$\lambda_i = g_i  \text{~for each~} v_i \in V$ {\color{winered}\tcc*[l]{local scaling factors}}
		{\color{winered}\tcc{Efficiently computing local scaling factors.}}
		\For{$c \in \rem$}{ \label{alg:loc-scalings-begin}
			$\money \gets \sum_{i \in N_c} (t - \spent_i)$ \;
			$\scaling \gets 0$ \;
			\For{$i \in N_c \text{~sorted in the increasing order of~} u_i(c)$}{
				$\mu \gets \money / \cost(c)$ \;
				$\scaling \gets \max(\scaling, \mu \cdot u_i(c))$ \;
				$\lambda_i \gets \max(\lambda_i,\scaling)$\;
				$\money \gets \money - (t - \spent_i)$\; \label{alg:loc-scalings-end}
			}
		} 
		{\color{winered}\tcc{This part is similar to the method of equal shares.}}
		\For{$c \in \rem$}{
			$x_i(c) \gets 2 \cdot (t - \spent_i) \cdot u_i(c) / (u_{i}(c) + \max(\lambda_i, u_i(c)))  \text{~for each~} v_i \in V$ \;
			\eIf{$\sum_{i \in N} x_i(c) < \cost(c)$}{
				$\rho(c) \gets +\infty$ {\color{winered}\tcc*[l]{candidate not affordable by what voters agree to pay}}
			}{
			$p \gets \cost(c)$ \;
			$u \gets \sum_{i \in N_c} u_i(c)$ \;
			\For{$i \in N_c \text{~sorted in the increasing order of~} x_i(c) / u_i(c)$}{
				\If{$x_i(c) \cdot p \cdot u_i(c)/u $}{
					$\rho(c) \gets p / u$ \;
					\textbf{break} {\color{winered}\tcc*[l]{we have found the optimal $\rho(c)$}}}
					$p \gets p - x_i(c)$\;
					$u \gets u - u_i(c)$\;
				}
			}
		}
		\eIf{$\rho(c) = \infty \text{~for all~} c \in \rem$}{ 
			$t \gets t + \Delta_t$ {\color{winered}\tcc*[l]{$\Delta_t$ is a parameter that controls the increase of time}}
		}{
			$c^* \gets \argmin_{c \in \rem} \rho(c)$ \;
			$W \gets W \cup \{c^*\}$ \;
			$\rem \gets \rem \setminus \big(\{c^*\} \cup \{ c\in \rem \colon W \cup \{c\} \notin \feasibles\}\big)$  \;
			{\color{winered}\tcc{Updating global scaling factors}}
			\For{$i \in N$}{
				$g_i \gets \max(g_i, \lambda)$\;
			}
			{\color{winered}\tcc{Updating budgets}}
			\For{$i \in N_{c^*}$}{
				$\spent_i \gets \spent_i + \min(x_i(c^*), \rho(c^*) \cdot u_i(c^*))$ \;
			}
		}
	}
	\Return{W}\;
\end{algorithm}

\begin{lemma}\label{lem:scaling_factor}
The value of the scaling factor will never be greater than $u_{\max}$.
\end{lemma}
\begin{proof}
For the sake of contradiction assume that this is not the case, and consider the first time~$t$ when the scaling factor for some voter becomes greater than $u_{\max}$. This means that at this time there must exist an $(\mu, S, \alpha)$-affordable candidate $c$ such that the new scaling factor becomes $\mu \alpha > u_{\max}$. In particular, $\mu > 1$. Without loss of generality, assume that $c$ is the candidate with the highest value of $\mu \alpha$ at the moment. 
For each voter $i \in S$ we have that:
\begin{align*}
x_i(c, r) &= \kappa \cdot \frac{2p_i(r)u_i(c)}{u_{i}(c) + \max(\lambda_i(r), u_i(c))} + (1 - \kappa) \cdot \frac{p_i(r)u_i(c)}{\max(\lambda_i(r), u_i(c))} \\
		  &> \kappa \cdot \frac{2p_i(r)u_i(c)}{u_{i}(c) + \mu u_i(c)} + (1 - \kappa) \cdot \frac{p_i(r)u_i(c)}{\mu u_i(c)} \\
		  &= \kappa \cdot \frac{2p_i(r)}{1 + \mu} + (1 - \kappa) \cdot p_i(r) > \kappa \cdot \frac{p_i(r)}{\mu} + (1 - \kappa) \cdot p_i(r) \\
		  &= \kappa \cdot \frac{p_i(r)}{\sum_{j \in S}p_j(r) / \cost(c)} + (1 - \kappa) \cdot p_i(r) \text{.}
\end{align*}
Since this candidate could not have been bought just before time $t$ it must be the case that:
\begin{align*}
\cost(c) &> \sum_{i \in S} x_i(c, r) > \kappa \cdot \sum_{i \in S} \frac{p_i(r)}{\sum_{j \in S}p_j(r) / \cost(c) } + (1 - \kappa) \sum_{i \in S} p_i(r) \\
         &= \kappa \cdot \cost(c) + (1 - \kappa) \sum_{i \in S} p_i(r)
\end{align*}
This means that:
\begin{align*}
\cost(c) > \sum_{i \in S} p_i(r) \text{.}
\end{align*}
From that we infer that $\mu < 1$, a contradiction.
\end{proof}

\Cref{lem:scaling_factor} gives us a certain intuitive estimation of how much a voter is willing to pay for a candidate $c$, namely $x_i(c, r)  \in [p_i(r)\cdot \nicefrac{u_i(c)}{u_{\max}}, p_i(r)]$. Interestingly, this means that for approval-based utilities $x_i(c, r) = p_i(r)$, and so we get the following corollary that holds for every $\kappa \in [0, 1]$.

\begin{corollary}
For for approval-based utilities PropRank becomes the Phragm\'en sequential rule.
\end{corollary}

\subsection{Proportionality Degree of PropRank}

We now move to discussing strong proportionality guarantees of the PropRank algorithm. We will present our main theoretical results, which essentially shows that when moving to additive utilities we do not loose any of the proportionality degree of the Phragm\'en sequential rule. However, first we will describe the strategy of the proof and provide certain useful lemmas. 

Consider an election $E = (N, C, \cost, \{u_i\}_{i \in N})$ and a committee $W$ returned by PropRank for~$E$. Further, consider an $\alpha$-cohesive group of voters $S$. We will estimate how much on average a voter from $S$ pays for a unit of satisfaction. For that we analyse the following potential function.
Similarly as in the definition of the rule, let $p_i(t)$ denote the amount of money held by voter $i \in S$ at time $t$; further, let $p_{\avg}(t) = \nicefrac{1}{|S|}\sum_{i \in S} p_i(t)$ denote the average amount of money held by the voters from $S$. For a time $t$ we define the potential value $\phi(t)$ as:
\begin{align*}
\phi(t) = \sum_{i \in S} \big(p_i(t) - p_{\avg}(t) \big)^2\text{.} 
\end{align*}
Roughly speaking, $\phi(t)$ is a variance of the balances of the voters in $S$. A similar potential function has been considered in the estimation of the proportionality degree of the Phragm\'{e}n's rule~\cite{skowron:prop-degree}.

First, observe that the potential value remains unchanged when the number of credits of each voter is incremented (earning credits does not change the potential value). Next, we analyse what happens when the voters use their money to pay for a candidate. The voters pay for a given candidate simultaneously, but for the sake of our analysis, we will split this into separate payments, one per each voter from $S$ (we will not consider the payments made by the voters outside of $S$).

Let $\lambda_{\mathrm{loc}}(c, S, t)$ be the local scaling factor resulting from considering candidate $c$ and a group $S$ only. Further, let $\lambda_{\mathrm{glo}}(S, t)$ be the highest value of such a local scaling factors considered at times $t' < t$, and over all candidates $c$ that were available at that time moment $t'$. Finally, let 
\begin{align*}
\lambda(c, S, t) = \max\left(\lambda_{\mathrm{glo}}(S, t), \lambda_{\mathrm{loc}}(c, S, t)\right) \text{.}
\end{align*}
Clearly, for a fixed subset $S$ the values $\lambda(c, S, t)$ are non-decreasing with respect to time $t$. Further, for each $i \in S$ the local scaling factor of $i$ at time $t$ is greater than or equal to $\lambda(c, S, t)$, for every candidate $c$ that is available at time $t$. Finally, we use the additional notation:
\begin{align*}
\gamma(c) = \frac{|S| \alpha(c)}{2\cost(c)} \text{.}
\end{align*}
Let $t(c)$ be the time when the first voter is about to make her first payment for $c$, and let $t^{\rightarrow}(c)$ be the time just after the purchase of $c$.

\begin{lemma}\label{lem:potential_change}
Assume that in $t(c)$ there was a not-yet elected candidate, call it $c_T$. Consider a purchase of a candidate, call it $c$, in which some voters from $S$ are involved. Take a voter $i \in S$ who paid $x$ dollars for $c$. Let $\Delta_{\phi}$ denote the change of the potential function due to such a single purchase of $i$, Then:
\begin{align}\label{eq:potential_change}
\Delta_{\phi} & \leq \lambda(c_T, S, t(c)) \cdot x \cdot \left(\frac{1}{\gamma(c_T)}  - \frac{x}{u_i(c)} \right) \text{.}
\end{align}
\end{lemma}
\begin{proof}
 We introduce an additional time $t(i, c)$ which formally equals to~$t(c)$, but corresponds to the moment where the voters who were paying before $i$ for $c$ already made their purchases, and~$i$ is about to make the payment. Additionally, let $t^{\rightarrow}(i, c)$ denote the time just after voter $i$ makes the payment for $c$ (formally, these two are of course equal to $t(c)$).

Let us assess the change of the potential function due to such a single purchase of $i$:
\begin{align*} 
\Delta_{\phi} &= \sum_{j \in S\colon j \neq i} \left(p_j(t(i, c)) - \left(p_{\avg}(t(i, c)) - \frac{x}{|S|}\right) \right)^2 + \left(p_i(t(i, c)) - x - \left(p_{\avg}(t(i, c)) - \frac{x}{|S|}\right) \right)^2 \\
                     &\qquad\qquad- \sum_{j \in S} \big(p_j(t(i, c)) - p_{\avg}(t(i, c)) \big)^2 \\
                     &= \sum_{j \in S} \left(p_j(t(i, c)) - \left(p_{\avg}(t(i, c)) - \frac{x}{|S|}\right) \right)^2 - \sum_{j \in S} \big(p_j(t(i, c)) - p_{\avg}(t(i, c)) \big)^2 \\
                     &\qquad\qquad + x^2 -2x\left(p_i(t(i, c))- p_{\avg}(t(i, c)) + \frac{x}{|S|}\right) \\
                     &= \sum_{j \in S}\frac{x}{|S|} \left(2p_j(t(i, c)) - 2p_{\avg}(t(i, c)) + \frac{x}{|S|} \right) + x^2 -2x\left(p_i(t(i, c))- p_{\avg}(t(i, c)) + \frac{x}{|S|}\right) \\
\end{align*} 
Since $\sum_{j \in S} \left( 2p_j(t(i, c)) - 2p_{\avg}(t(i, c))\right) = 0$, we get:
\begin{align*} 
\Delta_{\phi} &= |S| \cdot \frac{x}{|S|} \cdot \frac{x}{|S|} + x^2 -2x\left(p_i(t(i, c))- p_{\avg}(t(i, c)) + \frac{x}{|S|}\right) \\
                    &= x\left(2p_{\avg}(t(i, c)) - 2p_i(t(i, c))+  x - \frac{x}{|S|}\right) \leq x\left(2p_{\avg}(t(i, c)) - 2p_i(t(i, c))+  x\right)
\end{align*} 

Recall that in $t(c)$ there was a not-yet elected candidate $c_T$ available, and assume that that $c_T$ was $(\mu, S, \alpha)$-affordable. Further, we know that
\begin{align*}
p_{\avg}(t(i, c)) = \frac{1}{|S|} \cdot \sum_{j \in S} p_j(t(i, c)) \leq \frac{1}{|S|} \cdot \sum_{j \in S} p_j(t(c)) = \frac{\mu}{|S|} \cdot \cost(c_T) \leq \frac{\lambda(c_T, S, t(c))}{|S|\alpha} \cdot \cost(c_T)\text{.}
\end{align*}
Additionally, the value of $x$ is the highest when $\kappa = 1$, and so we know that:
\begin{align*}
- 2p_i(t(i, c))+  x &= - 2p_i(t(c))+  x = x \cdot \left( \frac{- 2p_i(t(c))}{x} + 1  \right) \\
                &\leq x \cdot \left( \frac{- 2p_i(t(c)) (u_{i}(c) + \max(\lambda_i(t(c)), u_i(c)))}{2p_i(t(c))u_i(c)} + 1  \right) \\
                &= x \cdot \frac{- 2p_i(t(c)) \max(\lambda_i(t(c)), u_i(c))}{2p_i(t(c))u_i(c)} \\
				&= - \frac{x}{u_i(c)}  \cdot \max(\lambda_i(t(c)), u_i(c)) \leq  - \frac{x}{u_i(c)}  \cdot \lambda(c_T, S, t(c)) \text{.}
\end{align*}
Thus, we can further estimate the change of the potential function:
\begin{align*} 
\Delta_{\phi} & \leq \lambda(c_T, S, t(c)) \cdot x \cdot \left(\frac{2\cost(c_T)}{|S|\alpha}  - \frac{x}{u_i(c)} \right) \text{.}
\end{align*}
This completes the proof of the lemma.
\end{proof}

Intuitively, \Cref{lem:potential_change} says that until $c_T$ is available all purchases for the candidates will be efficient, meaning that the payment per utility will be on average at most equal to $\nicefrac{1}{\gamma(c_T)}$. Below, we make this observation formal. 

\begin{lemma}\label{lem:util_estimation}
Consider a set of selected candidates $W$ and label $W = \{a_1, a_2, \ldots, a_r\}$ so that the indices of the candidates correspond to the order of them being selected. Fix the indices $p\leq r$, and let $c_T(a_j)$ denote the candidate $c_T$ with the highest value of $\gamma(c_T)$ among those that were available just before $a_j$ was bought. Let $x_i(c)$ denote the amount of money paid by voter $i$ for a candidate $c$. Then:
\begin{align*}
\sum_{j = 1}^p \sum_{i \in S}u_i(a_j) \geq \sum_{j = 1}^p \left( \gamma(c_T(a_j)) \cdot \sum_{i \in S}x_i(a_j) \right) \text{.}
\end{align*}
\end{lemma}
\begin{proof}
Let us use the notation from \Cref{lem:potential_change}.
We set $\zeta_{i, c} = \frac{u_i(c)}{x} - \gamma(c_T)$, and further rewrite Inequality~\eqref{eq:potential_change}:
\begin{align*}
\Delta_{\phi} &\leq \lambda(c_T, S, t(c)) \cdot x \cdot \left(\frac{1}{\gamma(c_T)}  - \frac{1}{ \gamma(c_T) + \zeta_{i, c}} \right) \\
              &= \lambda(c_T, S, t(c)) \cdot x \cdot \frac{\zeta_{i, c}}{ \gamma(c_T)(\gamma(c_T) + \zeta_{i, c})} \leq x \zeta_{i, c} \cdot \frac{\lambda(c_T, S, t(c))}{\gamma(c_T)^2} \text{.}
\end{align*}

Let us interpret the above estimation of $\Delta_{\phi}$ intuitively, and next we will make the reasoning formal. Until $c_T$ is selected, the voters on average get at least the utility of $\gamma(c_T)$ per dollar. Indeed, if a voter gets the utility lower by $|\zeta_{i, c}|$ per dollar (that is, $\zeta_{i, c}$ is negative), then for each such a dollar spent the potential function decreases by at least $|\zeta_{i, c}|\lambda_i(c_T, S, t(c))/\gamma(c_T)^2$. Conversely, if a voter gets the utility greater by $|\zeta_{i, c}|$ per dollar, then for each dollar spent the potential function will increase by at most $|\zeta_{i, c}| \lambda_i(c_T, S, t) /\gamma(c_T)^2$. Since the potential function is always non-negative, we will infer that on average the voters get at least $\gamma(c_T)$ per dollar. A similar reasoning can be applied to the case when $c_T$ is already selected, but $c_T'$ is not.

According to our notation, for each $c \in W$ we can write:
\begin{align*} 
\sum_{i \in S}u_i(c) &= \sum_{i \in S}\Big(x_i(c)\gamma(c_T) + x_i(c)\zeta_{i, c}\Big) = \gamma(c_T)\sum_{i \in S}x_i(c) + \sum_{i \in S}x_i(c)\zeta_{i, c} \\
&\geq \gamma(c_T)\sum_{i \in S}x_i(c) + \Big(\phi(t^{\rightarrow}(c)) - \phi(t(c))\Big) \cdot \frac{\gamma(c_T)^2}{\lambda(c_T, S, t(c))} \text{.}
\end{align*} 

By summing up over all $c \in \{a_1, \ldots, a_q\}$, and using the fact that coefficients $\frac{\gamma(c_T)^2}{\lambda(c_T, S, t(c))}$ are non-increasing over time, and that the potential function is non-negative, we get:
\begin{align*} 
&\sum_{j = 1}^p \sum_{i \in S}u_i(a_j) = \sum_{i \in S}u_i(a_1) + \sum_{j = 2}^p \sum_{i \in S}u_i(a_j) \\
 &\qquad \qquad \geq \gamma(c_T(a_1))\sum_{i \in S}x_i(a_1) + \Big(\phi(t^{\rightarrow}(a_1)) - \phi(t(a_1))\Big) \cdot \frac{\gamma(c_T(a_1))^2}{\lambda(c_T(a_1), S, t(a_1))} + \sum_{j = 2}^p \sum_{i \in S}u_i(a_j)\\
 &\qquad \qquad \geq \gamma(c_T(a_1))\sum_{i \in S}x_i(a_1) + \Big(\phi(t^{\rightarrow}(a_1)) - \phi(t(a_1))\Big) \cdot \frac{\gamma(c_T(a_2))^2}{\lambda(c_T(a_2), S, t(a_2))} + \sum_{j = 2}^p \sum_{i \in S}u_i(a_j)\\
 &\qquad \qquad \geq \ldots \geq \sum_{j = 1}^p \left( \gamma(c_T(a_j)) \cdot \sum_{i \in S}x_i(a_j) \right) + \Big(\phi(t^{\rightarrow}(a_p)) - \phi(t(a_1))\Big) \cdot \frac{\gamma(c_T(a_p))^2}{\lambda(c_T(a_p), S, t(a_p))}  \\
 &\qquad \qquad \geq \sum_{j = 1}^p \left( \gamma(c_T(a_j)) \cdot \sum_{i \in S}x_i(a_j) \right) \text{.}
\end{align*} 
This completes the proof.
\end{proof}

Now, we are ready to formulate our main theorems.

\begin{theorem}\label{thm:phrag_guarantee}
PropRank has the proportionality degree of $d(\alpha) = (\alpha - u_{\max})/2$ for participatory budgeting constraints.
\end{theorem}
\begin{proof}
Let $T$ be the set that witnesses that $S$ is $\alpha$-cohesive. Let us relabel $T = \{c_1, \ldots, c_z\}$ so that 
so that:
\begin{align*}
	\frac{\alpha(c_1)}{\cost(c_1)} \geq \frac{\alpha(c_2)}{\cost(c_2)} \geq \ldots \geq \frac{\alpha(c_z)}{\cost(c_z)} \text{.}
\end{align*}

Let $t = \nicefrac{b}{n}$, and let $f_S$ be a function that lower bounds how much utility the voters from $S$ get per given amount of money. From \Cref{lem:util_estimation} we get that:
\begin{enumerate}
	\item until $c_1$ is selected the function has the slope of at least $|S|\alpha(c_1) / 2\cost(c_1)$. Further, the function has this slope for the segment of arguments of size greater than or equal to $\cost(c_1)$.
	\item Then, until $c_2$ is selected $f_S$ has the slope of at least $|S|\alpha(c_2) / 2\cost(c_2)$. The function has this slope for the segment of arguments of size greater than or equal to $\cost(c_2)$. Etc.
\end{enumerate}
Clearly, $f_S$ is concave.
We need to only make sure that the candidates from $T$ have not been removed. Let $c_\mathrm{out}$ denote the candidate from $T$ with the highest ratio of $\frac{\alpha(c_\mathrm{out})}{\cost(c_\mathrm{out})}$ that had been removed. Let $t'$ denote the time of its removal. We will show that this must have happened when the voters spent at least $m_t$ dollars, where 
\begin{align*}
m_t = t |S| - \frac{\cost(c_\mathrm{out})u_{\max}}{\alpha(c_\mathrm{out})} \text{.}
\end{align*}
Let us assume, for the sake of contradiction, that this is not the case. First, we show that $t' \geq t$.
If this was not the case, then the voters at $t'$ in total had spent at most:
\begin{align*}
t \cdot n  - \frac{\cost(c_\mathrm{out})u_{\max}}{\alpha(c_\mathrm{out})} = b - \frac{\cost(c_\mathrm{out})u_{\max}}{\alpha(c_\mathrm{out})} \leq b - \cost(c_\mathrm{out})\text{.}
\end{align*}
Thus, $c_\mathrm{out}$ would not have been removed. Now, observe that just before $t'$ we had that:
\begin{align*}
\cost(c_\mathrm{out}) > \sum_{i \in S} x_i(c_\mathrm{out}, t) \geq \sum_{i \in S} \frac{p_i(r)\alpha(c_\mathrm{out})}{u_{\max}} \text{.}
\end{align*}
Thus, at time $t'$ the voters from $S$ have used at least the following amount of money:
\begin{align*}
 t |S| - \sum_{i \in S}p_i(t) \geq  t |S| -  \frac{\cost(c_\mathrm{out}) u_{\max}}{\alpha(c_\mathrm{out})} = m_t \text{.}
\end{align*}
This shows a contradiction, and so we get that $c_\mathrm{out}$ was removed after voters from $S$ spent at least $m_t$ dollars. Now, let $m_t'$ be the time when the last candidate with the index lower than $c_\mathrm{out}$ was selected or when $c_\mathrm{out}$ was removed, whichever is greater.

Now, let us estimate the value of $f_S$ at $m_t'$. Note that for arguments higher than $m_t'$ the slope of the function is at most $|S|\alpha(c_\mathrm{out}) / 2\cost(c_\mathrm{out})$.
Thus:
\begin{align*}
f(m_t') &\geq f(t |S|) - \frac{|S| \alpha(c_\mathrm{out})}{2\cost(c_\mathrm{out})} \cdot \frac{u_{\max} \cdot \cost(c_\mathrm{out})}{\alpha(c_\mathrm{out})} \geq  f(t |S|) - |S| \cdot \frac{u_{\max}}{2} \\
		&\geq f\left(\sum_{c \in T}\cost(c)\right) - |S| \cdot \frac{u_{\max}}{2} \geq \sum_{c \in T}\cost(c) \cdot \frac{|S|\alpha(c)}{2\cost(c)} - |S| \cdot\frac{u_{\max}}{2} = \frac{1}{2} \cdot |S| \cdot \left(\alpha - u_{\max}\right) \text{.}
\end{align*}
This completes the proof.
\end{proof}

Now, let us move to the case of general feasibility constraints.

\begin{theorem}\label{thm:phrag_guarantee_constraints}
Consider an outcome $W$ returned by the PropRank rule. For each $\alpha$-cohesive group of voters $S$, the average utility from $W$ within the group $S$ equals at least:
\begin{align*}
\frac{1}{|S|} \sum_{i \in S} u_i(W) \geq \alpha \cdot \frac{1 - \gamma}{2}  - \frac{ u_{\max}}{2} \text{.}
\end{align*}
where $\gamma = \nicefrac{|S|}{n}$, and $u_{\max}$ is the highest utility a voter assigns to a candidate.
\end{theorem}
\begin{proof}
Consider an $(\alpha, \beta)$-cohesive group of voters $S\subseteq N$, and the purchases made until time
\begin{align*}
	t < \frac{\beta}{|S| n} \cdot (n - |S|) \text{.}
\end{align*}
Let $T$ denote the set of candidates that have been purchased until $t$; clearly $|T| \leq tn$. Since $S$ is $(\alpha, \beta)$-cohesive we know that there exists a set $X$ such that $|X| \leq \beta$, $\sum_{c \in X} \min_{i \in S}u_i(c) \geq \alpha$, and $X \cup T \in \feasibles$. Let us rename the candidates in $X$ so that $X = \{c_1, c_2, \ldots, c_{|X|}\}$ and 
\begin{align*}
	\min_{i \in S}u_i(c_1) \geq \min_{i \in S}u_i(c_2) \geq \ldots \geq \min_{i \in S}u_i(c_{|X|}) \text{.}
\end{align*}
Let $c_z$ be the element of $X \setminus T$ with the lowest index (that is, with the highest value of $\min_{i \in S}u_i(c_z)$). Using \Cref{lem:util_estimation} we define the function $f$ that lower bounds how much the utility that the voters from $S$ get per given amount of money. 
\begin{align*}
 f(|X|) = \sum_{i = 1}^z \frac{|S| \cdot \alpha(c_i)}{2} + (m - z) \cdot \frac{|S| \cdot \alpha(c_z)}{2} \geq |S| \cdot \frac{\alpha}{2} \text{.}
\end{align*}
Similarly, as in the proof of \Cref{thm:phrag_guarantee} we infer that at $t$ the voters have at most $\nicefrac{u_{\max}}{\alpha(c_z)}$ money left. Thus, they spent at least:
\begin{align*}
m_t = t |S| - \frac{u_{\max}}{\alpha(c_z)} \text{.}
\end{align*}
Now, consider the moment when $c_z$ was removed or when the last candidate with the index lower than $z$ was purchased, whichever is later. Assume that the voters form $S$ spent $m_t'$ of their virtual budgets by then. Clearly, $m_t' \geq m_t$, and for arguments greater than $m_t'$ the slope of the function $f$ is at most $\nicefrac{|S|\alpha(c_z)}{2}$.
Since $f$ is concave, we get that the total utility garnered by the voters for spending $m_t'$ is at least:
\begin{align*}
f(t|S|) - \frac{|S|\alpha(c_z)}{2} \cdot \frac{u_{\max}}{\alpha(c_z)} \geq f(|X|) \cdot \frac{t|S|}{|X|} - \frac{|S|}{2} \cdot u_{\max} \geq |S| \frac{\alpha}{2} \cdot \frac{n - |S|}{n}  - \frac{|S|}{2} \cdot u_{\max}
\end{align*}
Thus, the average utility can be lower-bounded by
\begin{align*}
\frac{\alpha}{2} \cdot \frac{n - |S|}{n}  - \frac{ u_{\max}}{2}  \geq \frac{\alpha}{2} \cdot \frac{n - \gamma n}{n}  - \frac{ u_{\max}}{2} = \alpha \cdot \frac{1 - \gamma}{2}  - \frac{ u_{\max}}{2} \text{.}
\end{align*}
This completes the proof.
\end{proof}

Finally, \Cref{thm:phrag_guarantee_matroid} below generalizes the analogous result of \citet{mas-pie-sko:group-fairness} stated for the Phragm\'en sequential rule.

\begin{theorem}\label{thm:phrag_guarantee_matroid}
	PropRank has the weak proportionality degree of $d(\alpha) = (\alpha - u_{\max})/2$ for matroid constraints.
\end{theorem}
\begin{proof}
	Consider an $(\alpha, \beta)$-cohesive group of voters $S$.
	We first define the time $t$ as follows:
	\begin{align*}
	t = \frac{\beta}{|S|} + \frac{\Delta - 1}{n} \text{,}
	\end{align*}
	where $\Delta$ is the largest non-negative value in $[0, 1]$ such that at $t$ the voters from $S$ have at least $\Delta$ unspent dollars. Such a value is well defined since in particular $\Delta = 0$ satisfies the premise. There are two possibilities: (1) the voters from $S$ had exactly $\Delta$ unspent dollars at $t$ (it is possible that at $t$ there was made a purchase such that before the purchase the voters from $S$ had more than $\Delta$ unspent dollars, and after the purchase, they had less than $\Delta$ unspent dollars; in such a case, we simply consider the purchases---possibly fractional---made until they had exactly $\Delta$ unspent dollars), or (2) $\Delta = 1$, and the voters from $S$ have in total strictly more than $1$ dollar left.  

	Consider only the purchases made by the algorithm until time $t$, just before the last candidate was bought, that is when the voters still had more than $\Delta$ dollars. 
	Let $W$ denote the set of candidates bought until this time moment. From $W$ we remove $\beta-1$ candidates with the highest values of $\min_{i \in S}u_i(c)$, and let $T$ denote the set of candidates after such removals. We have that:
	\begin{align*}
	|T| &= |W| - (\beta - 1) < (t \cdot n - \Delta) - \beta + 1 = \frac{n \cdot \beta}{|S|} + \Delta - 1 - \Delta - \beta + 1 \\
	    &= \frac{n \cdot \beta}{|S|} - \beta  = \beta \cdot \frac{n - |S|}{|S|} \text{.}
	\end{align*}

	Thus, since $S$ is $(\alpha, \beta)$-cohesive, we know that there exists a set $X$ such that $|X| = \beta$, $|X| \cdot \min_{i \in S, c\in X}u_i(c) \geq \alpha$, and $X \cup T \in \feasibles$. Thus, $\min_{i \in S, c\in X}u_i(c) \geq \nicefrac{\alpha}{\beta}$.

	If $|T \cap X| \geq 1$, then we know that in $W$ there are at least $1 + (\beta - 1)$ candidates $c$ with $\min_{i \in S} u_i(c) \geq \nicefrac{\alpha}{\beta}$. This already gives our thesis. Otherwise, that is if $|T \cap X| = 0$ we know that:
	\begin{align*}
	|T \cup X| \geq |W| - (\beta - 1) + \beta = |W| + 1 > |W| \text{.}
	\end{align*}
	Given that $\feasibles$ form a matroid, we get that there is $c \in X$ such that $W \cup \{c\} \in \feasibles$. This means that before each purchase made until $t$, there always existed a not-removed candidate $c$ with $\min_{i \in S}u_i(c) \geq \nicefrac{\alpha}{\beta}$. Now, we can once again use \Cref{lem:util_estimation}, and get that the average utility per dollar for the purchases made by the voters from $S$ until time $t$ is at least $\nicefrac{|S|\alpha}{2\beta}$. 
	
	Let $m_t$ denote the amount of money that the voters from $S$ have used until time $t$. In case (1) we have that:
	\begin{align*}
		m_t = t |S| - \Delta = \beta + \frac{(\Delta - 1)|S|}{n} - \Delta \geq \beta - 1\text{.}
	\end{align*}
	In case (2), when $\Delta = 1$ we have that $t = \frac{\beta}{|S|}$, and 
	\begin{align*}
		\sum_{i \in S}p_i(t) \leq \sum_{i \in S}x_i(c, t)\cdot \frac{u_{\max}}{u_i(c)} \leq  \sum_{i \in S}x_i(c, t)\cdot \frac{u_{\max}\beta}{\alpha} \leq  \frac{u_{\max}\beta}{\alpha}
	\end{align*}
	so:
	\begin{align*}
		m_t = t |S| - \sum_{i \in S}p_i(t) \geq \beta - \frac{u_{\max}\cdot \beta}{\alpha} \text{.}
	\end{align*}
	In either case we get that 
	\begin{align*}
		m_t \geq \beta - \frac{u_{\max}\cdot \beta}{\alpha} \text{.}
	\end{align*}

	Consequently, we infer that the average utility gained by the voters from $S$ equals at least:
	\begin{align*}
	\frac{1}{|S|} \cdot m_t \cdot \frac{|S|\alpha}{2\beta} & \geq (\beta - \frac{u_{\max}\cdot \beta}{\alpha}) \cdot \frac{\alpha}{2\beta} \geq \frac{\alpha}{2} - \frac{u_{\max}}{2} \text{.}
	\end{align*}
	This completes the proof.
\end{proof}

\subsection{Proportional Justified Representation of PropRank}\label{pjr_degree}

A weaker notion of proportionality that is often considered is called Proportional Justified Representation (PJR). The axiom uses the same notion of cohesiveness, but the condition on the required voters' satisfaction is weaker. Below we formulate the quantitative variant of the axiom, which by analogy we call PJR degree.

\begin{definition}[Proportional Justified Representation Degree]\label{def:pjr_generalization}
	Consider an election $E = (N, C, \cost, \{u_i\}_{i \in N})$, and a subset of candidates $W \in \feasibles$. We say that $W$ has the \myemph{proportional justified representation degree (PJR degree)} of $d\colon \reals \to \reals$ if for each $\alpha \in \reals$, and $\alpha$-cohesive group of voters $S$ it holds that
	\begin{align*}
	\sum_{c \in W} \max_{i \in S} u_i(c) \geq d\left(\alpha\right) \text{.}
	\end{align*}
	A selection rule has the PJR degree of $d$ if it always returns outcomes with the PJR degree of $d$.  \hfill $\lrcorner$
\end{definition}

\begin{theorem}\label{thm:phrag_guarantee_pjr}
	PropRank with the parameter $\kappa$ has the PJR degree of $d(\alpha) = (\alpha - u_{\max})/(\kappa + 1)$ for participatory budgeting constraints.
\end{theorem}
\begin{proof}
	Consider an election $E = (N, C, \cost, \{u_i\}_{i \in N})$ and an outcome $W$ returned by PropRank for $E$. Consider an $\alpha$-cohesive group of voters $S$, and let $T$ denote the subset of candidates that witnesses that $S$ is $\alpha$-cohesive. For each candidate $c'$ we use the notation $\alpha(c') = \min_{i \in S} u_i(c')$. Let us relabel these candidates as $T = \{c_1, c_2, \ldots, c_z\}$ so that:
	\begin{align*}
		\frac{\alpha(c_1)}{\cost(c_1)} \geq \frac{\alpha(c_2)}{\cost(c_2)} \geq \ldots \geq \frac{\alpha(c_z)}{\cost(c_z)} \text{.}
	\end{align*}
	For each time $t$ let $m(t)$ denote the maximal amount of money that some voter from $S$ has spent until $t$. Consider a purchase that changed this value from $m(t)$ to $m(t')$, which is witnessed by a voter $i \in S$. The purchased happened in time $t$, and before the purchase the voter had spent $m_i$ dollars; clearly $m_i \leq m(t)$. Assume that at this moment the candidate $c_T \in T$ has not yet been selected. Then the local scaling factor is lower-bounded by:
	\begin{align*}
		\lambda \geq \frac{|S|(t - m(t))}{\cost(c_T)} \cdot \alpha(c_T)
	\end{align*} 
	Thus, the satisfaction of some voter in $S$ due to the purchase increased by the value $u$ such that
	\begin{align*}
		m(t') - m_i &\leq \kappa \cdot \frac{2 (t - m_i)u}{u + \max(\lambda, u)} + (1 - \kappa) \cdot \frac{(t - m_i)u}{\max(\lambda, u)} \\
		             & \leq \frac{(2\kappa + 1 - \kappa) \cdot (t - m_i) \cdot u}{\max(\lambda, u)} \leq \frac{(\kappa + 1) \cdot (t - m_i) \cdot u}{\lambda} \\
					 & \leq \frac{(\kappa + 1) \cdot (t - m_i) \cdot u \cdot \cost(c_T)}{ \alpha(c_T) |S|(t - m(t))} \text{.}
	\end{align*} 
	After reformulation we get that:
	\begin{align*}
		\frac{u}{m(t') - m(t)} \geq \frac{1}{\kappa + 1} \cdot \frac{\alpha(c_T)|S|}{\cost(c_T)} \cdot \frac{(m(t') - m_i) \cdot (t - m(t))}{(m(t') - m(t)) \cdot (t - m_i)} \geq \frac{1}{\kappa + 1} \cdot \frac{\alpha(c_T)|S|}{\cost(c_T)} \text{.}
	\end{align*}

	Now, as in the previous proofs we consider a function $f$ that given the value of $m(t)$ lower bounds the satisfaction from the condition of PJR (namely: $\sum_{c \in W'} \max_{i \in S} u_i(c)$ ) that can be obtained with the purchases such that no voter pays more than $m(t)$. We infer that for $x = (\kappa + 1) \cdot \frac{\cost(c_1)}{|S|}$ the value of the function is at least $\alpha(c_1)$. Similarly, for $x = (\kappa + 1) \cdot \left(\frac{\cost(c_1)}{|S|} + \frac{\cost(c_2)}{|S|}\right)$ the value of the function is at least $\alpha(c_1) + \alpha(c_2)$. Further, note that:
	\begin{align*}
		\sum_{c_T \in T}(\kappa + 1) \frac{\cost(c_T)}{|S|} \leq \sum_{c_T \in T} (\kappa + 1) \cdot \cost(c_T) \cdot \frac{b}{\cost(T) \cdot n} = (\kappa + 1) \cdot \frac{b}{n} \text{.}
	\end{align*}
	Thus, the value of the function for $x = \nicefrac{b}{n}$ equals at least $\nicefrac{\alpha}{\kappa + 1}$.
	
	Consider the first time when a candidate $c_T \in T$ was removed. We will show that some voter must have spent at least $\nicefrac{b}{n} - \frac{\cost(c_T)u_{\max}}{|S|\alpha(c_T)}$ money when this happened. For the sake of contradiction, assume that this is not the case. We will first show that $c_T$ could not have been removed at time $t' < \nicefrac{b}{n}$. Indeed, the total money spent at $t'$ would be at most:
	\begin{align*}
		t'(n - |S|) + |S| \cdot \left(\nicefrac{b}{n} - \frac{\cost(c_T)u_{\max}}{|S|\alpha(c_T)}\right) &< \nicefrac{b}{n} \cdot (n - |S|) + |S| \cdot \left(\nicefrac{b}{n} - \frac{\cost(c_T)u_{\max}}{|S|\alpha(c_T)}\right) \\
		&= b - \frac{\cost(c_T)u_{\max}}{\alpha(c_T)} \leq b - \cost(c_T) \text{.}
	\end{align*}
	This gives a contradiction, and so $c_T$ must had been removed at least at time $t'' = \nicefrac{b}{n}$. At this time it must had been the case that:
	\begin{align*}
		\cost(c_T) &> \sum_{i \in S} \kappa \cdot \frac{2 (t'' - m(t''))u_i(c_T)}{u_i(c_T) + \max(\lambda, u_i(c_T))} + (1 - \kappa) \cdot \frac{(t'' - m(t''))u_i(c_T)}{\max(\lambda, u_i(c_T))} \\
		           &\geq \sum_{i \in S} \cdot \frac{(t'' - m(t''))u_i(c_T)}{\max(\lambda, u_i(c_T))} \geq |S| \cdot \frac{(t'' - m(t''))\alpha(c_T)}{u_{\max}} \text{.}
	\end{align*} 
	After reformulation:
	\begin{align*}
		t'' - m(t'') < \frac{\cost(c_T)u_{\max}}{|S|\alpha(c_T)} \text{.}
	\end{align*}
	Which shows that some voter must have spent at least $\nicefrac{b}{n} - \frac{\cost(c_T)u_{\max}}{|S|\alpha(c_T)}$ money. Let $\tau$ be the maximal sepnding of the voter from $S$ when the last candidate with the index lower than $c_T$ was purchased or when $c_T$ was removed, whichever is greater. Clearly $\tau \geq \nicefrac{b}{n} - \frac{\cost(c_T)u_{\max}}{|S|\alpha(c_T)}$ and the function $f$ for arguments greater or equal than $\tau$ has a scope of at most $\frac{1}{\kappa + 1} \cdot \frac{\alpha(c_T)|S|}{\cost(c_T)}$.

	Thus, we can assess the value of the function $f$ for $\tau$ as:
	\begin{align*}
		f(\tau) \geq f(\nicefrac{b}{n}) - \frac{\cost(c_T)u_{\max}}{|S|\alpha(c_T)} \cdot \frac{1}{\kappa + 1} \cdot \frac{\alpha(c_T)|S|}{\cost(c_T)} \geq \frac{\alpha}{\kappa + 1} - u_{\max} \cdot \frac{1}{\kappa + 1} \text{.}
	\end{align*}
	This completes the proof.
\end{proof}

\begin{corollary}
	PropRank with the parameter $\kappa = 0$ satisfies PJR up to the maximal utility.
\end{corollary}

\Cref{thm:phrag_guarantee_pjr} suggests that the best value of $\kappa$ is $0$, in which case the algorithm simplifies and becomes somewhat more intuitive. However, this is not the value we recommend, as our experiments provide strong evidence in favor of higher values. The intuitive explanation is that a larger $\kappa$ gives voters more flexibility to purchase candidates and reduces the amount of money left unused at the end of the procedure. Hence, we believe that $\kappa = 1$ should be the defaul version. 

Note that an analogous results to \Cref{thm:phrag_guarantee_pjr} apply to other types of constraints, and the proof follows the same idea as the analogous proofs of \Cref{thm:phrag_guarantee_constraints} and \Cref{thm:phrag_guarantee_matroid}.

\begin{theorem}\label{thm:phrag_pjr_guarantee_constraints}
Consider an outcome $W$ returned by the PropRank rule. For each $\alpha$-cohesive group of voters $S$, it holds that
\begin{align*}
\sum_{c \in W} \max_{i \in S} u_i(c) \geq \alpha \cdot \frac{1 - \gamma}{\kappa + 1}  - \frac{ u_{\max}}{\kappa + 1} \text{.}
\end{align*}
where $\gamma = \nicefrac{|S|}{n}$, and $u_{\max}$ is the highest utility a voter assigns to a candidate.
\end{theorem}

\begin{theorem}\label{thm:phrag_pjr_guarantee_matroid}
	PropRank has the weak PJR degree of $d(\alpha) = (\alpha - u_{\max})/(\kappa + 1)$ for matroid constraints.
\end{theorem}

\subsection{Remark on Global Scaling Factors}

Observe that the proof of \Cref{thm:phrag_guarantee_pjr} does not rely on the monotonicity of the scaling factors. Therefore, if we are only interested in the PJR degree, the rule can be simplified by removing the global scaling factors. We will use this observation also later, when designing the variant of the method of equal shares for general constraints.

\section{Method of Equal Shares for General Constraints}

Now, we build on the ideas from \Cref{sec:prop_rank} to design an extension of the method of equal shares~\cite{pet-sko:laminar,pet-pie-sko:c:participatory-budgeting-cardinal} to the setting with general feasibility constraints. Let us start by emphasizing the main difference between two well-established voting rules for committee elections: Phragm{'e}n's rule~\cite{aaai/BrillFJL17-phragmen, lac-sko:multiwinner-book} and the method of equal shares. While both are based on the idea of voters buying approved candidates, they differ in one important aspect. Phragm{'e}n's rule distributes money gradually over time, with purchases made greedily as soon as they become affordable. In contrast, Equal Shares allocates fixed budgets to voters upfront, which enables more purchases with a good payment-per-utility ratio and ultimately leads to stronger proportionality guarantees. The PropRank algorithm, similarly to Phragm{'e}n's rule, also proceeds greedily but incorporates a foresight mechanism: it evaluates whether a purchase is sufficiently profitable or whether it is better to delay spending in anticipation of more valuable candidates becoming affordable later. However, PropRank itself does not implement the idea of providing voters with sufficient upfront budgets to enable multiple purchases within a given allocation of money.

Our generalized method of equal shares combines the ideas of PropRank and the original method of equal shares. It can be viewed as a sequence of invocations of the Equal Shares ADD1 subroutine, where in each invocation voters are given as much money as possible upfront. However, within each invocation we may deliberately postpone purchasing candidates if we anticipate that more valuable ones will become available in later invocations of the subroutine---just as in PropRank.

\subsection{The Main Loop of the Algorithm}

In the main algorithm we seek the largest amount of money that can be allocated to the voters such that the Equal Shares subroutine purchases only candidates that respect the feasibility constraints. After executing these purchases, we remove candidates that are no longer feasible given the current partial solution and repeat the process. In rare cases, it may happen that the largest feasible allocation of money does not allow for any purchase. If we were to increase the voters' entitlements by an arbitrarily small amount, the subroutine would then purchase candidates that, taken together, violate the feasibility constraints. In this case, we first identify the set of candidates that would be purchased under such a slightly increased entitlement, and then select from this set a feasible subset. This guarantees that each invocation of the Equal Shares subroutine adds at least one candidate to the outcome, hence the termination of the algorithm. 

The pseudo-code of the main algorithm is given in \Cref{alg:equal-shares-general}. 

\subsection{The Equal Shares with a Foresight Subroutine}

As we have noted, the overall rule can be seen as a sequence of invocations of the Equal Shares ADD1 subroutine. However, in each invocation we may deliberately postpone purchasing  candidates if we expect that more valuable ones may become available in future iterations---just as in PropRank. This is done by computing the upper bounds on spending on particular candidates. Compared to PropRank, where such upper bounds are given by Equality~\eqref{eq:caps}, in case of the Equal Shares in round $r$ they are defined as follows. Assuming that $W$ is the set of already selected candidates, we first compute the best possible payment-per-utility if all the voters would had enough money:
\begin{align*}
	\rho_{\best}(i) = \min_{c \in \rem} \frac{\cost(c)}{\sum_{j \in N} u_j(c)} \text{~~where~~} \rem = \{c \in C \setminus W \colon W \cup \{c\} \in \feasibles \} \text{.}
\end{align*}
Then, we compute the upper bound $\mathrm{cap}_i$ as:
\begin{align}\label{eq:mes-general-equation}
	\begin{split}
		&\mathrm{cap}_i(c, r) = \max(x_i(c, r), y_i(c, r)) \text{~~where:} \\
		&x_i(c, r) = \kappa \cdot \frac{2p_i(r)u_i(c)}{u_{i}(c) + \max(\lambda_i(r), u_i(c))} + (1 - \kappa) \cdot \frac{p_i(r)u_i(c)}{\max(\lambda_i(r), u_i(c))} \text{,} \\
		&y_i(c, r) = \min(\rho_{\best}(i) \cdot u_i(c) \cdot (1 + \kappa), p_i(r)).
	\end{split}
\end{align}

This equation is similar to Equality~\eqref{eq:caps} for PropRank. The key difference is that, unlike in PropRank, the initial voter budgets $p_i$ can be large. As a result, at the early stages we may be able to cover substantial fractions of candidates' costs---often fractions exceeding $1$. Consequently, the local scaling factors $\lambda$ may become very large, which in turn drives the values of $x_i(c, r)$ down. For this reason, $x_i(c, r)$ can be larger than $y_i(c, r)$, a situation that does not arise in the original PropRank. Equality~\Cref{eq:mes-general-equation} accounts for this effect. To minimize the amount of unused budget, it also takes $y_i(c, r)$ into consideration, allowing $\mathrm{cap}_i(c, r)$ to be as large as possible while still ensuring that payments per unit of utility remain efficient.

Below we present an analogous lemma to \Cref{lem:scaling_factor} with (almost) the same proof.

\begin{lemma}\label{lem:scaling_factor_mes}
At the end of each iteration of MES the local scaling factors are no greater than $u_{\max}$.
\end{lemma}
\begin{proof}
For the sake of contradiction assume that this is not the case. This means that when all the purchases are made, there must exist an $(\mu, S, \alpha)$-affordable candidate $c$ such that the new scaling factor becomes $\mu \alpha > u_{\max}$. In particular, $\mu > 1$. Without loss of generality, assume that $c$ is the candidate with the highest value of $\mu \alpha$ at the moment. 
For each voter $i \in S$ we have that:
\begin{align*}
x_i(c, r) &= \kappa \cdot \frac{2p_i(r)u_i(c)}{u_{i}(c) + \max(\lambda_i(r), u_i(c))} + (1 - \kappa) \cdot \frac{p_i(r)u_i(c)}{\max(\lambda_i(r), u_i(c))} \\
		  &> \kappa \cdot \frac{2p_i(r)u_i(c)}{u_{i}(c) + \mu u_i(c)} + (1 - \kappa) \cdot \frac{p_i(r)u_i(c)}{\mu u_i(c)} \\
		  &= \kappa \cdot \frac{2p_i(r)}{1 + \mu} + (1 - \kappa) \cdot p_i(r) > \kappa \cdot \frac{p_i(r)}{\mu} + (1 - \kappa) \cdot p_i(r) \\
		  &= \kappa \cdot \frac{p_i(r)}{\sum_{j \in S}p_j(r) / \cost(c)} + (1 - \kappa) \cdot p_i(r) \text{.}
\end{align*}
Since $c$ could not have been bought (at the end of the iteration all the purchases are alerady made) it must be the case that:
\begin{align*}
\cost(c) &> \sum_{i \in S} x_i(c, r) > \kappa \cdot \sum_{i \in S} \frac{p_i(r)}{\sum_{j \in S}p_j(r) / \cost(c) } + (1 - \kappa) \sum_{i \in S} p_i(r) \\
         &= \kappa \cdot \cost(c) + (1 - \kappa) \sum_{i \in S} p_i(r)
\end{align*}
This means that:
\begin{align*}
\cost(c) > \sum_{i \in S} p_i(r) \text{.}
\end{align*}
From that we infer that $\mu < 1$, a contradiction.
\end{proof}

\SetKw{Return}{return}
\SetKw{Input}{Input:}
\begin{algorithm}[!thb]
	\small
	\captionsetup{labelfont={sc,bf}, labelsep=newline}
	\caption{Pseudo-code of the method of equal shares for general constraints.}\label{alg:equal-shares-general}
	\Input{an election with feasibiity constraints $\feasibles$\\}
	\DontPrintSemicolon
	\SetAlgoNoEnd
	\SetAlgoLined
	\SetKwFunction{FMES}{Equal-Shares-Subroutine}
    \SetKwProg{Fn}{Function}{:}{}

	$W \gets \emptyset$ \; 
	$\rem \gets C$ {\color{winered}\tcc*[l]{remaining candidates}}
	$\spent_i \gets 0  \text{~for each~} v_i \in V$  {\color{winered}\tcc*[l]{money spent by each voter}}
	\While{$\rem \neq \emptyset$}{
		{\color{winered}\tcc{Formally, the subroutine returns a set of candidates and payments, but below we treat the outcome as simply the set of candidates.}}
		$t_{\max} \gets \text{largest value}~t~\text{such that \FMES}(\spent, t, \rem) \cup W \in \feasibles$\;
		$W_{\nxt}, p_{\nxt} \gets \text{\FMES}(\spent, t_{\max}, \rem)$\;
		\If{$W_{\nxt} = \emptyset$}{
			$W_{\nxt}', p_{\nxt} \gets \text{\FMES}(\spent, t_{\max} + \epsilon, \rem)$\;
			$W_{\nxt} \gets \text{~largest subset of~}W_{\nxt}' \text{~such that~}W_{\nxt} \cup W \in \feasibles$\;
		}
		$W \gets W \cup W_{\nxt}$ \; 
		$\spent_i \gets \spent_i + \sum_{c \in W_{\nxt}} p_{\nxt}[v_i][c]$\;
		$\rem \gets \{ c \in \rem \setminus W \colon W \cup \{c\} \in \feasibles\}$\;
	}
	\Return{W}\;
	\;
	\Fn{\FMES{$\spent$, $t$, $\rem$}}{
		$W \gets \emptyset$ \; 
		$\textrm{pay}[v_i][c] \gets 0 \text{~for all~} v_i \in V, c \in C$ \; 
		\While{$\rem \neq \emptyset$}{
			$\lambda_i = 0  \text{~for each~} v_i \in V$ {\color{winered}\tcc*[l]{local scaling factors}}
			$\text{Compute efficiently local scaling factors (lines \ref{alg:loc-scalings-begin}--\ref{alg:loc-scalings-end} in \Cref{alg:base_prop_rank_algorithm})}$\;
			{\color{winered}\tcc{This part is similar to the method of equal shares.}}
			\For{$c \in \rem$}{
				$\text{compute~}\mathrm{cap}_i(c) \text{~for all~} v_i \in V, c \in \rem \text{~according to \eqref{eq:mes-general-equation} using}~\lambda_i$\;
				$x_i(c) \gets \min(\mathrm{cap}_i(c), t -  \spent_i)  \text{~for each~} v_i \in V$ \;
				\eIf{$\sum_{i \in N} x_i(c) < \cost(c)$}{
					$\rho(c) \gets +\infty$ {\color{winered}\tcc*[l]{candidate not affordable by what voters agree to pay}}
				}{
				$p \gets \cost(c)$ \;
				$u \gets \sum_{i \in N_c} u_i(c)$ \;
				\For{$i \in N_c \text{~sorted in the increasing order of~} x_i(c) / u_i(c)$}{
					\If{$x_i(c) \cdot p \cdot u_i(c)/u $}{
						$\rho(c) \gets p / u$ \;
						\textbf{break} {\color{winered}\tcc*[l]{we have found the optimal $\rho(c)$}}}
						$p \gets p - x_i(c)$\;
						$u \gets u - u_i(c)$\;
					}
				}
			}
			\eIf{$\rho(c) = \infty \text{~for all~} c \in \rem$}{ 
				\Return{$W, \textrm{pay}$}\;
			}{
				$c^* \gets \argmin_{c \in \rem} \rho(c)$ \;
				$W \gets W \cup \{c^*\}$ \;
				$\rem \gets \rem \setminus \{c^*\}$  \;
				\For{$i \in N_{c^*}$}{
					$\textrm{pay}[v_i][c^*] \gets \min(x_i(c^*), \rho(c^*) \cdot u_i(c^*))$ \;
					$\spent_i \gets \spent_i + \textrm{pay}[v_i][c^*]$ \;
				}
			}
		}
    }
	\;
\end{algorithm}

\subsection{PJR Degree of Equal Shares with General Constraints}

Note that the estimation of the potential function in the proofs of \Cref{thm:phrag_guarantee,thm:phrag_guarantee_constraints,thm:phrag_guarantee_matroid} remains unchanged, since it does not rely on the specific timestamps at which purchases occur. However, because all purchases now occur at much later timestamps (compared to the standard PropRank), the global scaling factors may be set to excessively high values. As a result, voters might wait far too long before making purchases in subsequent rounds. In other words, in these later rounds voters could maintain unreasonably high savings, potentially leading to unbalanced satisfaction, even though the purchases they make would still exhibit good payment-per-utility ratios. In order to prevent that, we resign from global scaling factors. As it is explained in \Cref{pjr_degree} this variant preserves the PJR degree guarantees, though it is hard to formally reason about the stronger property of proportionality degree. The following theorem preserves with the same proof as in the case of the PropRank algorithm.

\begin{theorem}\label{thm:mes_pjr}
	The method of equal shares with the parameter $\kappa$ has the PJR degree of $d(\alpha) = (\alpha - u_{\max})/(\kappa + 1)$ for participatory budgeting constraints.
\end{theorem}
\begin{proof}
For each time $t$ let $m(t)$ denote the maximal amount of money that some voter from $S$ has spent until $t$. Consider a purchase that changed this value from $m(t)$ to $m(t')$, , which is witnessed by a voter $i \in S$, and assume that at this moment the candidate $c_T \in T$ has not yet been selected. Then we have two cases. If $\mathrm{cap}_i(c, r) = x_i(c, r)$, then we repeat the caluculations from the proof of \Cref{thm:phrag_guarantee_pjr}, and get that the satisfaction of some voter in $S$ due to the purchase increased by the value $u$ such that:
\begin{align*}
	\frac{u}{m(t') - m(t)} \geq \frac{1}{\kappa + 1} \cdot \frac{\alpha(c_T)|S|}{\cost(c_T)} \text{.}
\end{align*}

Otherwise, that is if $\mathrm{cap}_i(c, r) = y_i(c, r)$, then we get that:
\begin{align*}
	\frac{u}{m(t') - m(t)} &\geq \frac{u_i(c)}{y_i(c, r)} = \frac{u_i(c)}{\rho_{\best}(i) \cdot u_i(c) \cdot (\kappa + 1)} = \frac{1}{\rho_{\best}(i) \cdot (\kappa + 1)} \\
						   &\geq \frac{\sum_{j \in N} u_j(c_T)}{\cost(c_T) \cdot (\kappa + 1)} \geq \frac{1}{\kappa + 1} \cdot \frac{\alpha(c_T)|S|}{\cost(c_T)} \text{.}
\end{align*}

It remains to show that when $c_T$ was removed then at least one voter from $S$ must have spent at least $\nicefrac{b}{n} - \frac{\cost(c_T)u_{\max}}{|S|\alpha(c_T)}$. If we prove this, the proof will continue the same way as the proof of \Cref{thm:phrag_guarantee_pjr}. For the sake of contradiction, assume that this is not the case. We will first show that $c_T$ could not have been removed at time $t' < \nicefrac{b}{n}$. 

Indeed, the total money spent at $t'$ would be at most:
	\begin{align*}
		t'(n - |S|) + |S| \cdot \left(\nicefrac{b}{n} - \frac{\cost(c_T)u_{\max}}{|S|\alpha(c_T)}\right) &< \nicefrac{b}{n} \cdot (n - |S|) + |S| \cdot \left(\nicefrac{b}{n} - \frac{\cost(c_T)u_{\max}}{|S|\alpha(c_T)}\right) \\
		&= b - \frac{\cost(c_T)u_{\max}}{\alpha(c_T)} \leq b - \cost(c_T) \text{.}
	\end{align*}
	This gives a contradiction, and so $c_T$ must had been removed at least at time $t'' = \nicefrac{b}{n}$. At this time it must had been the case that:
	\begin{align*}
		\cost(c_T) &> \sum_{i \in S} \kappa \cdot \frac{2 (t'' - m(t''))u_i(c_T)}{u_i(c_T) + \max(\lambda, u_i(c_T))} + (1 - \kappa) \cdot \frac{(t'' - m(t''))u_i(c_T)}{\max(\lambda, u_i(c_T))} \\
		           &\geq \sum_{i \in S} \cdot \frac{(t'' - m(t''))u_i(c_T)}{\max(\lambda, u_i(c_T))} \geq |S| \cdot \frac{(t'' - m(t''))\alpha(c_T)}{u_{\max}} \text{.}
	\end{align*} 
	In the above estimation we use the bound on the local scaling factors, which we can do since the candidates are removed only at the end of an iteration, that is after an invocation of a single subrouting. 

	After reformulation:
	\begin{align*}
		t'' - m(t'') < \frac{\cost(c_T)u_{\max}}{|S|\alpha(c_T)} \text{.}
	\end{align*}
	Which shows that some voter must have spent at least $\nicefrac{b}{n} - \frac{\cost(c_T)u_{\max}}{|S|\alpha(c_T)}$ money. This completes the proof.
\end{proof}

The proofs of the next two theorems follow exactly the same strategies as for the case of PropRank. We provide one of them with details as an illustration.

\begin{theorem}\label{thm:mes_guarantee_constraints}
	Consider an outcome $W$ returned by the method of equal shares with parameter~$\kappa$. For each $(\alpha, \beta)$-cohesive group of voters $S$, it holds that:
	\begin{align*}
	\sum_{c \in W} \max_{i \in S} u_i(c)  \geq \alpha \cdot \frac{1 - \gamma}{\kappa + 1}  - \frac{ u_{\max}}{\kappa + 1} \text{.}
	\end{align*}
	where $\gamma = \nicefrac{|S|}{n}$, and $u_{\max}$ is the highest utility a voter assigns to a candidate.
\end{theorem}
\begin{proof}
	The proof combines the ideas from the proofs of \Cref{thm:phrag_guarantee_constraints} and \Cref{thm:mes_pjr}. Consider an $(\alpha, \beta)$-cohesive group of voters $S\subseteq N$, and the purchases made until $t$:
	\begin{align*}
		t = \frac{\beta}{|S| n} \cdot (n - |S|) \text{.}
	\end{align*}
	Let $T$ be the set of candidates purchased until $t$; clearly $|T| \leq tn$. For each candidate $c'$ we use the notation $\alpha(c') = \min_{i \in S} u_i(c')$. Since $S$ is $(\alpha, \beta)$-cohesive we know that there exists a set $X$ such that $|X| \leq \beta$, $\sum_{c \in X} \alpha(c) \geq \alpha$, and $X \cup T \in \feasibles$. Let us rename the candidates in $X$ so that $X = \{c_1, c_2, \ldots, c_{|X|}\}$ and 
	\begin{align*}
		\alpha(c_1) \geq \alpha(c_2) \geq \ldots \geq \alpha(c_{|X|}) \text{.}
	\end{align*}
	Let $c_z$ be the element of $X \setminus T$ with the lowest index (that is, with the highest value of $\alpha(c_z)$).
	
	For each time $t$ let $m(t)$ denote the maximal amount of money that some voter from $S$ has spent until $t$. Consider a purchase that changed this value from $m(t)$ to $m(t')$, , which is witnessed by a voter $i \in S$, and assume that at this moment the candidate $c_X \in X$ has not yet been selected. We repeat the caluculations from the proof of \Cref{thm:mes_pjr}, and get that the satisfaction of some voter in $S$ due to the purchase increased by the value $u$ such that:
	\begin{align*}
		\frac{u}{m(t') - m(t)} \geq \frac{1}{\kappa + 1} \cdot\alpha(c_X)|S| \text{.}
	\end{align*}

	Now, as before we consider a function $f$ that given the value of $m(t)$ lower bounds the satisfaction from the condition of PJR that can be obtained with the purchases such that no voter pays more than $m(t)$. We infer that for $x = \frac{\kappa + 1}{|S|}$ the value of the function is at least $\alpha(c_1)$. Similarly, for $x = \frac{2(\kappa + 1)}{|S|}$  the value of the function is at least $\alpha(c_1) + \alpha(c_2)$, and for  $x = \frac{\beta(\kappa + 1)}{|S|}$ its value is at least $\alpha$. Thus, for $x = t$ the function has the value of at least
	\begin{align*}
		\frac{\alpha \cdot t|S|}{\beta(\kappa + 1)} = \frac{\beta}{|S| n} \cdot (n - |S|) \cdot \frac{\alpha \cdot|S|}{\beta(\kappa + 1)} = \frac{n - |S|}{n} \cdot \frac{\alpha}{\kappa + 1}\text{.}
	\end{align*}

	Similarly, as in the proof of \Cref{thm:mes_pjr} we infer that when $c_z$ was removed the voters spent at least:
	\begin{align*}
	m_t = t - \frac{u_{\max}}{|S|\alpha(c_z)} \text{.}
	\end{align*}
	Now, consider the moment when $c_z$ was removed or when the last candidate with the index lower than $z$ was purchased, whichever is later. Assume that some voter form $S$ spent $m_t'$ of their virtual budgets by then. Clearly, $m_t' \geq m_t$, and for arguments greater than $m_t'$ the slope of the function $f$ is at most $\nicefrac{|S|\alpha(c_z)}{\kappa + 1}$.
	Since $f$ is concave, we get that:
	\begin{align*}
	f(m_t') \geq f(t) - \frac{|S|\alpha(c_z)}{\kappa + 1} \cdot \frac{u_{\max}}{\alpha(c_z)|S|} \geq  \frac{n - |S|}{n} \cdot \frac{\alpha}{\kappa + 1} - \frac{u_{\max}}{\kappa + 1}
	\end{align*}
	This completes the proof.
\end{proof}

\begin{theorem}\label{thm:phrag_pjr_guarantee_matroid}
	The method of equal shares with parameter $\kappa$ has the weak PJR degree of $d(\alpha) = (\alpha - u_{\max})/(\kappa + 1)$ for matroid constraints.
\end{theorem}

\section{Equal Shares with Bounded Overspending and Backtrack Heuristics}\label{sec:heuristics}

In extensive experiments on the PabuLib datasets~\cite{fal-fli-pet-pie-sko-sto-szu-tal:pb-experiments}, we observed that the primary source of suboptimal proportionality in the PropRank algorithm is that voters often end up with substantial amounts of unspent budget (captured in our theoretical analysis by the $u_{\max}$ parameter). Indeed, our experimental results show that PropRank performs significantly better for $\kappa = 1$ than for $\kappa = 0$. Motivated by this observation, we propose two heuristic algorithms that remain closely aligned with the core design principles of PropRank:

\begin{description}
	\item[PropRankRem.] In the PropRank algorithm, the local scalings are computed by taking
the minimum  of $\mu \cdot u_i(c)$ over all available candidates
(lines~\ref{alg:loc-scalings-begin}--\ref{alg:loc-scalings-end}
in~\Cref{alg:base_prop_rank_algorithm}). Intuitively, these are the candidates that may yield higher utility and are therefore worth waiting for. 
However, this
conservative behavior can lead to voters retaining large amounts of unused
budget by the end of the algorithm.
PropRankRem addresses this issue by restarting the procedure whenever a
candidate $c_{\mathrm{rem}}$ is removed from consideration. Intuitively, if such a candidate is not selected, then the voters should not keep their money waiting for it being affordable in the future. After
restarting, the local scaling factors are computed the same way as before, but now the removed
candidate $c_{\mathrm{rem}}$ is excluded from the set over which the
minimum is taken. This means that voters who were previously holding back
their budget in anticipation of $c_{\mathrm{rem}}$ are no longer influenced
by it. Since each restart removes exactly one candidate from
consideration, the PropRank algorithm is repeated at most $m$ times, with
the set of candidates used in the computation of local scalings shrinking by one
in each restarted run.

\item[PropRankBacktrack.] This algorithm is similar in spirit to PropRank, but performs backtracking in a more controlled manner. It consists of multiple executions of PropRank. In each execution, local scaling factors are computed using only candidates from a designated waiting set, which we denote as $C_{\mathrm{wait}}$. Initially, we set $C_{\mathrm{wait}} := C$.

After each run, we update $C_{\mathrm{wait}}$ as follows. First, we add to $C_{\mathrm{wait}}$ all candidates selected during the previous run. Then, we remove from $C_{\mathrm{wait}}$ the candidate that was removed earliest during that previous run. If no such candidate exists---that is, if $C_{\mathrm{wait}}$ is already a subset of the previously selected set $W_{\mathrm{prev}}$---the procedure terminates and returns $W_{\mathrm{prev}}$. Otherwise, PropRank is executed again with the updated $C_{\mathrm{wait}}$.

During a new execution of PropRank, we retain the previous execution for reference. Let $C_{\mathrm{prev}}$ and $C_{\mathrm{new}}$ denote the waiting sets used for computing local scalings in the previous and current executions, respectively. At each timestamp, we compare which candidate would be selected under each of these two sets, that is, when local scalings are computed with respect to $C_{\mathrm{prev}}$ or $C_{\mathrm{new}}$.
If a candidate would be selected under $C_{\mathrm{prev}}$, we follow the previous execution and select that candidate. If instead a candidate would be selected under $C_{\mathrm{new}}$ but no candidate would be selected under $C_{\mathrm{prev}}$ at that time, we select it and from that point onward continue using only $C_{\mathrm{new}}$ to compute local scalings, discarding the previous execution entirely.

In other words, the new execution initially mirrors the previous one. However, it may eventually reach a timestamp at which a candidate would be selected under $C_{\mathrm{new}}$, while no candidate was selected at that time in the previous run. This candidate is then accepted, and from that moment on the previous execution is no longer considered.

The procedure terminates because in each successive run, either the vector of timestamps at which candidates are purchased decreases lexicographically, or the set of candidates used for computing local scalings strictly shrinks. Nevertheless, the procedure may require exponential time. To address this, we introduce an additional parameter $\sigma$ that controls the running time. Specifically, each candidate may be reintroduced to $C_{\mathrm{wait}}$ at most $\sigma$ times. Whenever a candidate that is not currently in $C_{\mathrm{wait}}$ is added back (because it was selected in the previous execution), we increment a counter associated with that candidate. Once the counter reaches $\sigma$, the candidate is no longer reintroduced to $C_{\mathrm{wait}}$. Consequently, the procedure consists of at most $\sigma m$ runs of PropRank.

\item[Bounded Overspending (BOS).] We additionally explain how to adapt the recently introduced Method of Equal Shares with Bounded Overspending (BOS)~\cite{pap-pis-ski-sko-was:bos} to the setting with general feasibility constraints. The method follows the same high-level structure as MES: at each step, it identifies a timestamp, increases the voters’ monetary entitlements accordingly, performs purchases using the original BOS procedure, and removes candidates whose inclusion would violate feasibility.

The main component that requires further specification is how to determine the timestamp at which the next invocation of BOS should take place. Consider a single execution of BOS. If, during such an execution, BOS attempts to select a candidate that is infeasible, that candidate is simply skipped, as in the standard BOS procedure. However, if BOS would attempt to select an infeasible candidate that could be fully funded without overspending, we additionally declare this entire execution of BOS invalid. We then define the timestamp as the largest value for which the corresponding execution of BOS is valid.

We provide the pseudo-code of BOS for general constraints in \Cref{alg:bounded-overspending}.
\end{description}

\SetKw{Return}{return}
\SetKw{Input}{Input:}
\begin{algorithm}[t]
	\small
	\captionsetup{labelfont={sc,bf}, labelsep=newline}
	\caption{Pseudo-code of the Method of Equal Shares with Bounded Overspending.}\label{alg:bounded-overspending}
	\Input{an election with feasibiity constraints $\feasibles$\\}
	\DontPrintSemicolon
	\SetAlgoNoEnd
	\SetAlgoLined
	\SetKwFunction{FBOS}{BOS-Subroutine}
    \SetKwProg{Fn}{Function}{:}{}

	$W \gets \emptyset$ \; 
	$\rem \gets C$ {\color{winered}\tcc*[l]{remaining candidates}}
	$\spent_i \gets 0  \text{~for each~} v_i \in V$  {\color{winered}\tcc*[l]{money spent by each voter}}
	\While{$\rem \neq \emptyset$}{
		$t_{\max} \gets \text{largest}~t~\text{such that \FBOS}(\spent, t, \rem)~\text{returns True as last coordinate}$\;
		$W_{\nxt}, p_{\nxt}, \_ \gets \text{\FBOS}(\spent, t_{\max}, \rem)$\;
		$W \gets W \cup W_{\nxt}$ \; 
		$\spent_i \gets \spent_i + \sum_{c \in W_{\nxt}} p_{\nxt}[v_i][c]$\;
		$\rem \gets \{ c \in \rem \setminus W \colon W \cup \{c\} \in \feasibles\}$\;
	}
	\Return{W}\;
	\;
	\Fn{\FBOS{$\spent$, $t$, $\rem$}}{
	$\textrm{pay}[v_i][c] \gets 0 \text{~for all~} v_i \in V, c \in C$ \; 
	$b_i \gets t - \spent_i  \text{~for each~} v_i \in V$ \;
	\While{\emph{exists} $c \in C \setminus W$ \emph{such that} $W \cup \{c\} \in \feasibles$ \emph{and} $\sum_{v_i \in V: u_i(c) > 0} b_i > 0$}{
		$(\alpha^*, \rho^*, c^*) \gets (1, +\infty, c)$ \;
		
		\For{$c \in C \setminus W$ \emph{such that} $W \cup \{c\} \in \feasibles$}{
			$\rho' \gets \rho$ satisfying $\cost(c) = \sum_{i=1}^n \min(b_i, u_i(c) \cdot \rho)$ \;
			\For{$\rho \in \{b_i/u_i(c) : v_i \in V, b_i >0, u_i(c) > 0\} \cup \{\rho'\}$}{
				$\alpha \gets \min(\left(\sum_{i=1}^n \min(b_{i},u_{i}(c) \cdot \rho)\right)/\cost(c),1)$ \;
				\If{$\rho/\alpha < \rho^*/\alpha^*$}{
					$(\alpha^*, \rho^*, c^*) \gets (\alpha, \rho, c)$ \;
				}
			}
		}
		\For{$c \in C \setminus W$ \emph{such that} $W \cup \{c\} \notin \feasibles$ \emph{and} $\sum_{v_i \in V: u_i(c) > 0} b_i > 0$}{
			$\rho' \gets \rho$ satisfying $\cost(c) = \sum_{i=1}^n \min(b_i, u_i(c) \cdot \rho)$ \;
			\If{$\rho' < \rho^*/\alpha^*$}{
				\Return{$\emptyset$, $\mathrm{pay}$, $\mathrm{False}$}\;
			}
		}
		$W \gets W \cup \{c^*\}$ \;
		\For{$v_i \in V$ \emph{such that} $b_i > 0$ \emph{and} $u_i(c^*) > 0$}{
			$b_i \gets \max(0, b_i - u_i(c^*) \cdot \rho^*)$ \;	
		}
	}
	\Return{$W$, $\mathrm{pay}$, $\mathrm{True}$}\;
	}
\end{algorithm}

\section{Experimental Evaluation}

In this section, we present an experimental evaluation of our rules using data from real participatory budgeting instances~\cite{fal-fli-pet-pie-sko-sto-szu-tal:pb-experiments}, encompassing over 1,300 instances. Following standard assumptions in the literature, we adopt the cost-utility framework. Specifically, for elections with approval ballots, if a voter approves a candidate $c$, she assigns that candidate a utility of $\cost(c)$; otherwise, she assigns zero utility.

We evaluate the instances using two fairness metrics and one utilitarian efficiency metric:
\begin{description}
\item[Extended Justified Representation Plus Violations (EJR+)]~\cite{BP-ejrplus}:
This metric measures violations of one of the strongest proportionality axioms studied in the literature---a slightly stronger variant of EJR that is verifiable in polynomial time. This serves as our primary evaluation metric. We use the up-to-one variant of EJR+ adapted to cost-utilities~\cite{BP-ejrplus}.

\item[Exclusion Ratio (ER)]~\cite{fal-fli-pet-pie-sko-sto-szu-tal:pb-experiments}:
This metric represents the fraction of voters who are completely excluded, meaning none of the candidates they support receive funding:
\begin{align*}
\frac{1}{n} \cdot \left|\left\{v_i \in N : \sum_{c \in W} u_i(c) = 0\right\}\right|.
\end{align*}

\item[Average Cost-Satisfaction (Cost-Sat)]~\cite{fal-fli-pet-pie-sko-sto-szu-tal:pb-experiments}:
This metric represents the average total utility received by voters:
\begin{align*}
\frac{1}{n} \cdot \sum_{v_i \in N} \sum_{c \in W} u_i(c) \text{.}
\end{align*}
To enable cross-election comparisons, we normalize this value by dividing by the corresponding metric value obtained from the greedy rule used in the actual election, following \citet{pap-pis-ski-sko-was:bos}.
\end{description}

We also compare our rules with three well-established rules tailored to participatory budgeting: the method of equal shares with ADD1 completion (MES-PB)~\cite{pet-sko:laminar, pet-pie-sko:c:participatory-budgeting-cardinal}, the method of equal shares with bounded overspending (BOS-PB)~\cite{pap-pis-ski-sko-was:bos}, and the greedy method (Greedy-PB)~\cite{fal-fli-pet-pie-sko-sto-szu-tal:pb-experiments}, which simply selects candidates in order of the total utility they received from voters, skipping those whose cost would exceed the remaining budget given the current selection.

\begin{table*}[t!bh!]
\caption{Average cost satisfaction, exclusion ratio, and average number of EJR+ violations of the outputs produced by the considered rules, evaluated on \pabulib{} data as a function of the number of projects in an instance. Cost satisfaction is reported as a fraction of the corresponding satisfaction under the utilitarian rule. }
\label{table:statistics}
\begin{center}
\footnotesize
\begin{tabular}{ c || c | c | c || c | c | c || c | c | c } 
 & \multicolumn{3}{c||}{Cost-Sat} & \multicolumn{3}{c||}{Exclusion ratio} & \multicolumn{3}{c}{EJR+} \\
 \toprule
 instance size & $\leq 10$ & $11$ -- $30$ & $> 30$ & $\leq 10$ & $11$ -- $30$ & $> 30$ & $\leq 10$ & $11$ -- $30$ & $> 30$ \\
\toprule
Greedy  &  1 &1 & 1                     & 20.6\% & 19.06\% & 19.9\%              & 0.17 & 0.57  & 2.9\\ 
\midrule
MES-PB  &  0.82 & 0.82 & 0.84             & 23\%   & 16.3\% & 13.7\%                 & 0 & 0  & 0\\ 
BOS-PB  &  0.94 & 0.87 & 0.87             & 19.8\% & 14.8\% & 12.7\%                & 0.1 & 0.027  & 0.032\\ 
  
\midrule
PropRank ($\kappa = 0$)  &  0.93 & 0.9 &0.91             & 21.2\% & 16.1\% & 14.4\%             & 0.12 & 0.2  & 0.15\\        
PropRank ($\kappa = 1$)  &  0.82 & 0.8 & 0.83             & 22\% & 15.6\% & 13.4\%                & 0.005 & 0  & 0\\   

\midrule
MES ($\kappa = 0$)  &  0.92 & 0.9 &0.91             & 21.2\% & 16.1\% & 14.4\%             & 0.12 & 0.2  & 0.15\\        
MES ($\kappa = 1$)  &  0.86 & 0.8 & 0.83             & 21.8\% & 15.6\% & 13.4\%                & 0.005 & 0  & 0\\   

\midrule
PropRankRem ($\kappa = 0$) &  0.84 & 0.87 & 0.83             & 21.5\% & 15.4\% & 12.9\%             & 0.005 & 0.034  & 0\\
PropRankRem ($\kappa = 1$) &  0.81 & 0.8 & 0.78             & 22.1\% & 15.6\% & 12.9\%             & 0 & 0  & 0\\

\midrule
BOS  &  0.92 & 0.89 & 0.88             & 21\% & 15.2\% & 13.3\%             & 0.033 & 0.01  & 0.005\\

\end{tabular}
\end{center}
\end{table*}

We draw the following conclusions:
\begin{enumerate}
    \item First, we observe that the PropRank algorithm with $\kappa = 1$ achieves substantially better proportionality than the variant with $\kappa = 0$.
    \item The heuristic algorithms discussed in \Cref{sec:heuristics} yield further improvements over the guarantees of the basic PropRank variant, even though we cannot prove the analogous theoretical guarantees for these rules. While the base PropRank algorithm performs best with $\kappa = 1$, our experiments indicate that combining it with the mechanisms from \Cref{sec:heuristics} makes $\kappa = 0$ preferable. In this configuration, we observe very few EJR+ violations and favorable exclusion ratio values, while maintaining total utility comparable to that of the best algorithms specifically designed for participatory budgeting.
    \item The variant of the Method of Equal Shares with Bounded Overspending, adapted to general constraints, performs particularly well---even outperforming the original BOS-PB variant in terms of EJR+ violations and performing similarly in terms of exclusion ratio and total utility. This is especially noteworthy given that this rule is also the most computationally efficient among the techniques discussed in \Cref{sec:heuristics}.
\end{enumerate}

\section{Conclusion}

We have developed proportional algorithms for a general election model with arbitrary feasibility constraints and additive voter utilities. In particular, we demonstrated how established methods from the committee election literature---namely the Phragm\'en sequential rule and the Method of Equal Shares, both with and without bounded overspending---can be extended to this broader setting.
These extensions required new conceptual tools, most notably a mechanism for determining when voters should spend their virtual budgets on currently affordable candidates rather than postponing their spending in anticipation of more preferred candidates becoming affordable in later rounds.

Since our PropRank algorithm satisfies committee monotonicity, it can be used for constructing proportional rankings of the candidates. 

We proved that our rules satisfy strong proportionality guarantees. At the same time, we found that the underlying ideas can be pushed further toward more advanced heuristic variants, albeit at the cost of increased running time. Although the proof techniques used for the base algorithms do not directly extend to these heuristics, the heuristics remain firmly grounded in the same conceptual principles as the rules with provable guarantees.
Our experimental results confirm that the proposed rules perform well with respect to proportionality metrics when applied to participatory budgeting instances. In particular, the heuristic variants perform exceptionally well, producing virtually no EJR+ violations.

\subsection*{Acknowledgements}
We thanks Dominik Peters and Jannik Peters for the most valuable discussions over the ideas presented in the paper. 

The author was supported by the European Union (ERC, PRO-DEMOCRATIC, 101076570). Views and opinions expressed are however those of the authors only and do not necessarily reflect those of the European Union or the European Research Council. Neither the European Union nor the granting authority can be held responsible for them.

\begin{center}
\includegraphics*[scale=0.22]{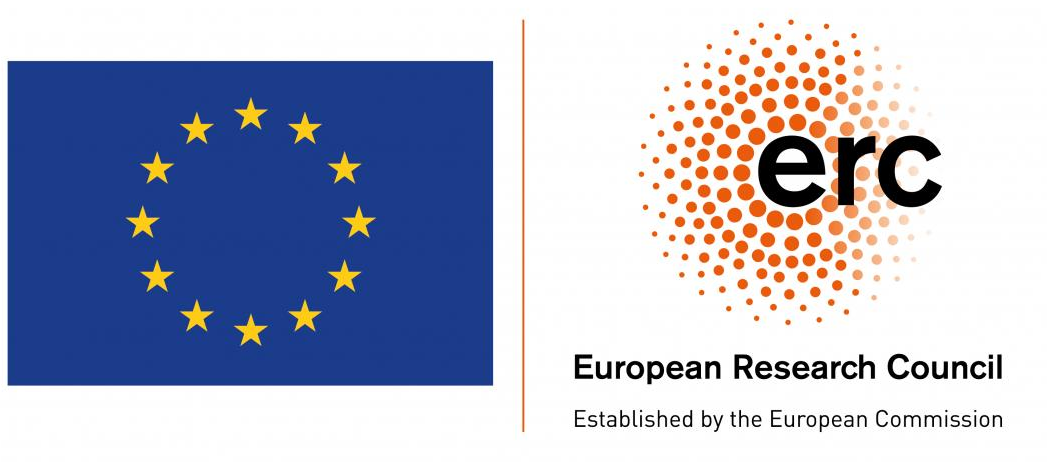}	
\end{center}

\bibliographystyle{plainnat}
\bibliography{main}

\end{document}